\documentclass[ukenglish,notitlepage,a4paper,12pt]{article}
\usepackage{natbib}
\usepackage{array}
\usepackage{multirow}
\usepackage{amsthm}
\usepackage{amsmath}
\usepackage{amsfonts}
\usepackage{amssymb}
\usepackage{pgf}
\usepackage{bm}
\usepackage{paralist}
\usepackage{pdflscape}
\usepackage{rotating}
\usepackage{scalefnt}
\usepackage{xr}
\usepackage{xr-hyper}
\usepackage{verbatim}
\usepackage[flushleft]{threeparttable}
\usepackage[bottom]{footmisc}
\usepackage{amsmath}
\usepackage{graphicx}
\usepackage{epstopdf}
\usepackage{setspace}
\usepackage{booktabs}
\usepackage{multirow}
\usepackage{ccaption}
\usepackage{amsfonts}
\usepackage{url}
\usepackage{rotating}
\usepackage{colortbl}
\usepackage{lscape}
\usepackage{xcolor}
\usepackage{amssymb}
\usepackage{float}
\usepackage[font=small,labelfont=bf]{caption}
\usepackage{titlesec}
\usepackage{subcaption}
\usepackage{amsmath}
\usepackage{makeidx}
\usepackage{scalefnt}
\usepackage{epstopdf}
\usepackage{epsfig}
\usepackage{sectsty}
\usepackage[flushleft]{threeparttable}
\usepackage{array}
\usepackage{arydshln}
\usepackage{natbib}
\usepackage{titling}
\usepackage{xpatch}
\usepackage{amsthm}
\usepackage[normalem]{ulem}
\usepackage{array}
\usepackage{natbib}
\usepackage{geometry}
\usepackage{graphicx}
\usepackage{array}
\usepackage{lmodern}
\usepackage{multirow}
\usepackage{amsthm}
\usepackage{amsmath}
\usepackage{amsfonts}
\usepackage{amssymb}
\usepackage{adjustbox}
\usepackage{pgf}
\usepackage{bm}
\usepackage{paralist}
\usepackage{float}
\usepackage{pdflscape}
\usepackage{rotating}
\usepackage{scalefnt}
\usepackage{xr}
\usepackage{verbatim}
\usepackage[flushleft]{threeparttable}
\usepackage{graphicx}
\usepackage{chngpage}
\usepackage[skip=20pt]{caption}
\usepackage{enumitem}
\usepackage{caption}
\usepackage{subcaption}
\usepackage{lmodern,blindtext}
\usepackage{indentfirst}
\usepackage[colorlinks,urlcolor = black,pdfpagelabels,pdfstartview = FitH,
  bookmarksopen = true,bookmarksnumbered = true,linkcolor = blue,
  plainpages = false,hypertexnames = true,citecolor = blue,
]{hyperref}

\newtheorem{theorem}{Theorem}
\newtheorem{proposition}{Proposition}

\newtheorem{assumpB}{Assumption}
\newtheorem{assumpC}{Assumption}

\newtheorem{corollary}{Corollary}

\renewcommand{\thefootnote}{\fnsymbol{footnote}}
\newcommand{\ignore}[1]{}

\newcommand{\beq}{\begin{equation}}
\newcommand{\eeq}{\end{equation}}
\newcommand{\bea}{\begin{eqnarray}}
\newcommand{\eea}{\end{eqnarray}}
\newcommand{\bit}{\begin{itemize}}
\newcommand{\eit}{\end{itemize}}
\newcommand{\ben}{\begin{enumerate}}
\newcommand{\een}{\end{enumerate}}
\newcommand{\bpm}{\begin{pmatrix}}
\newcommand{\epm}{\end{pmatrix}}
\newcommand{\bbm}{\begin{bmatrix}}
\newcommand{\ebm}{\end{bmatrix}}

\titleformat{\section}{\Large\bfseries}{\thesection}{1em}{}

\font\myfont=cmr12 at 16pt
\font\myfontb=cmr12 at 10pt

\begin{document}

\title{\textbf{\myfont Vector or Matrix Factor Model? A Strong Rule Helps!}}
\author{ \myfontb Yong He\thanks{%
Shandong University, Email:heyong@sdu.edu.cn}, Xin-Bing Kong\thanks{%
Nanjing Audit University, Email:xinbingkong@126.com}, Lorenzo Trapani
\thanks{%
University of Nottingham, Email:Lorenzo.Trapani@nottingham.ac.uk}, Long Yu%
\thanks{%
National University of Singapore, Email:stayl@nus.edu.sg} }
\date{}
\maketitle

%EndAName
%\myfontb Shandong University, Nanjing Audit University,\\
%\myfontb University of Nottingham, National University of Singapore }

%Testing the number of factors in large matrix time series}\\

\begin{center}
\begin{minipage}{150mm}
			
			%\baselineskip=0.55 true cm
			
		\footnotesize{	{\it Abstract}:This paper investigates the issue of determining the dimensions of row and column factor spaces in matrix-valued data. Exploiting the eigen-gap in the spectrum of sample second moment matrices of the data, we propose a family of randomised tests to check whether a one-way or two-way factor structure exists or not. Our tests do not require any arbitrary thresholding on the eigenvalues, and can be applied with no restrictions on the relative rate of divergence of the cross-sections to the sample sizes as they pass to infinity. Although tests are based on a randomization which does not vanish asymptotically, we propose a de-randomized, \textquotedblleft strong\textquotedblright\ (based on the Law of the Iterated Logarithm) decision rule to choose in favour or against the presence of common factors. We use the proposed tests and decision rule in two ways. We further cast our individual tests in a sequential procedure whose output is an estimate of the number of common factors. Our tests are built on two variants of the sample second moment matrix of the data: one based on a row (or column) \textquotedblleft flattened\textquotedblright\ version of the matrix-valued sequence, and one based on a projection-based method. Our simulations show that both procedures work well in large samples and, in small samples, the one based on the projection method delivers a superior performance compared to existing methods in virtually all cases considered.

			\bigskip
			
			
				}		
		\end{minipage}
\end{center}

%\baselineskip=0.70 true cm

%\baselineskip=0.50 true cm
\doublespacing
\renewcommand*{\thefootnote}{\arabic{footnote}}

\section{Introduction}

% Matrix time series with factor structure

Matrix time series can be defined as a sequence of $p_{1}\times p_{2}$\
random matrices $\left\{ X_{t},1\leq t\leq T\right\} $, with each random
matrix used to model observations that are well structured to be an array.
Such datasets are of great interest in a wide variety of applied sciences in
general, and in social sciences in particular. For example, in
macroeconomics a \textquotedblleft classical\textquotedblright\ application
of matrix-valued time series (see the recent paper by %
\citealp{chen2021factor}, and the discussion therein) is modelling the
import-export volumes between countries for one product family such as e.g.
chemical, food, or machinery and electronic. In this example, also known as
a \textquotedblleft dynamic transport network\textquotedblright , at each
point in time one can construct a matrix where the columns represent imports
into a country and the rows exports towards a country (with the main
diagonal of course empty). Another possible example, studied in %
\citet{wang2019factor}, is a matrix of time series whose rows contain some
macroeconomic indicator (GDP, inflation, interest rates...) and whose
columns represent different countries. Further, in the context of financial
data, \citet{wang2019factor} study a matrix-valued time-series of portfolio
returns where each portfolio is identified by a size level and by a BE ratio
level; in the same paper, another example is provided, conceptually similar
to the one based on macroeconomic indicators described above, where a
matrix-valued time series is considered with different companies on each
row, and different company financials on each column. Finally, in marketing
studies, a very promising application is to time series of customers'
ratings on a large number of items in an online platform; a well-known
application of such a \textquotedblleft recommender
system\textquotedblright\ (\citealp{koren})\ is where, as time elapses,
several customers are asked to express their level of satisfaction with
several movies/TV shows. We also refer to the papers by \citet{fan2021} and %
\citet{Gao2021A} for further discussion and examples ranging from health
sciences (such as electronic health records and ICU data), to 2-D image data
processing.

When dealing with such complex datasets, exploring the possibility of
dimensionality reduction is of pivotal importance. A possible way of
achieving this is to \textquotedblleft flatten\textquotedblright\ the data,
and model the vectorised sequence of matrices as%
\begin{equation}
\underset{p_{1}p_{2}\times 1}{\text{Vec}\left( X_{t}\right) }=\underset{%
p_{1}p_{2}\times k}{\Lambda }\underset{k\times 1}{f_{t}}+\underset{%
p_{1}p_{2}\times 1}{u_{t}},  \label{bai03}
\end{equation}%
where $f_{t}$ is a (low-dimensional) vector of common factors. Such a
modelling strategy has been studied, in the context of vector-valued series,
in numerous contributions, and we refer to \citet{bai2016econometric} for a
comprehensive review. Although (\ref{bai03}) does lead to dimension
reduction, further refinements may still be desirable. On the one hand, (\ref%
{bai03}) requires the estimation of $p_{1}p_{2}k$ parameters in the loading
matrix $\Lambda $; this number may still be too large in empirical
applications, in particular when the cross-sectional dimensions $p_{1}$ and $%
p_{2}$ are large. On the other hand, given that $X_{t}$ is a sequence of
matrices, a better modelling approach could be based on allowing for the
presence of common factors along the rows and along the columns of $X_{t}$,
rather than destroying the matrix nature of the data by vectorising it.

\subsection{The two-way factor model for matrix-valued time series\label%
{sub21}}

In order to make full use of the matrix structure, a parsimonious modelling
approach has been proposed in a recent, seminal paper by %
\citet{wang2019factor}, who assume that $X_{t}$ is driven by a
low-dimensional set of common factors across the row and column dimensions:
\begin{equation}
\underset{p_{1}\times p_{2}}{X_{t}}=\underset{p_{1}\times k_{1}}{R}\underset{%
k_{1}\times k_{2}}{F_{t}}\underset{k_{2}\times p_{2}}{C^{\prime }}+\underset{%
p_{1}\times p_{2}}{E_{t}},\ 1\leq t\leq T,\ k_{1},\ k_{2}>0.  \label{fm}
\end{equation}%
In (\ref{fm}), $R$ is the $p_{1}\times k_{1}$ row factor loading matrix
explaining the variation of $X_{t}$ across the rows, $C$ is the $p_{2}\times
k_{2}$ column factor loading matrix reflecting the differences across the
columns of $X_{t}$, $F_{t}$ is the common factor matrix, and $E_{t}$ is an
idiosyncratic component. At a glance, a natural competitor of (\ref{fm})
could be a group (vector) factor model (see e.g. \citealp{ando}, and %
\citealp{mircorubin}): in such a class of models, there is only one
cross-section, and this one cross-section contains variables of the same
nature (say, considering an example above, the set of macroeconomic
indicators) which are well-grouped with known or unknown group membership.
The common factors are organised into groups, and the interrelations within
and between groups are characterised by such factors. Conversely, the data $%
X_{t}$ in (\ref{fm}) are genuinely matrix-valued, with two cross-sectional
dimensions of different nature (considering the examples mentioned above,
these could be countries and macroeconomic indicators in the context of
macroeconomic data; or customers and commodities in recommending systems).
Hence, the common components in the matrix factor models reflect the
interplay between the two different cross-sections: for example, in the
context of recommending systems, ratings are high whenever the purchasers'
consumption preferences{\ (rows in $R$)} match the underlying
characteristics of items displayed online{\ (rows in $C$), thus (\ref{fm})
is a natural modeling of the interactive effect between the row and column
cross sections.} In this context, it is natural to expect that (\ref{fm}),
which takes the matrix nature of the data into account, is a better approach
than using models based on vectorising $X_{t}$, where the presence of groups
arises from artificially stacking the columns (rows) of matrix-valued data.
Moreover, (\ref{fm}) has the added bonus of reducing the dimensionality
compared to a model like (\ref{bai03}): whilst in the latter case one needs
to estimate $p_{1}p_{2}k$ coefficients, in the case of (\ref{fm}) such
parameter complexity is reduced to $p_{1}k_{1}+p_{2}k_{2}$. This can be
viewed even more neatly if one considers the following alternative version
of (\ref{fm})%
\begin{equation}
\text{Vec}\left( X_{t}\right) =\left( C\otimes R\right) \text{Vec}\left(
F_{t}\right) +\text{Vec}\left( E_{t}\right) ,  \label{kron}
\end{equation}%
where \textquotedblleft $\otimes $\textquotedblright\ denotes the Kronecker
product. Equation (\ref{kron}) shows that the loadings associated with the
factor structure in $X_{t}$ satisfy a Kronecker product structure, whence
the higher parsimony of (\ref{fm}).

\bigskip

In order to better understand the nature of (\ref{kron}), we consider the
following example, where - as also mentioned above - $X_{t}$ represents a $%
p_{1}\times p_{2}$ time series whose colums contain some macroeconomic
indicators (GDP, inflation, interest rates...), and whose rows represent
different countries.\footnote{%
In Section \ref{example} of the Supplementary Material, we also discuss
another example, based on a similar discussion in {\cite{wang2019factor}},
which also illustrates the relationship between (\ref{kron}) and a
multilevel factor model.} {Consider the following notation: $R=(\alpha
_{p_{1}\times r},\widetilde{R}_{p_{1}\times (k_{1}-r-l)},\mathbf{1}%
_{p_{1}\times l})$, $C=(\mathbf{1}_{p_{2}\times r},\widetilde{C}%
_{p_{2}\times (k_{2}-r-l)},\beta _{p_{2}\times l})$ and $F_{t}=\mbox{diag}%
\{(g_{t})_{r\times r},(\widetilde{F}_{t})_{(k_{1}-r-l)\times
(k_{2}-r-l)},(h_{t})_{l\times l}\}$. Then, (\ref{kron}) becomes%
\begin{equation}
X_{t}=\alpha g_{t}\mathbf{1}_{r\times p_{2}}+\mathbf{1}_{p_{1}\times
l}h_{t}\beta ^{\prime }+\widetilde{R}\widetilde{F}_{t}\widetilde{C}^{\prime
}+E_{t}.  \label{ex1}
\end{equation}%
In (\ref{ex1}), $g_{t}$ and $h_{t}$ are the common factors along the row and
column cross-sections, respectively, and $\alpha $ and $\beta $ represent
their loadings; }$\widetilde{R}\widetilde{F}_{t}\widetilde{C}^{\prime }${\
may be viewed as an interaction effect component. Model (\ref{ex1}) can be
rewritten in vector form, for both the countries $j=1,...,p_{1}$ and the
indicators $i=1,...,p_{2}$:
\begin{eqnarray}
X_{j\cdot ,t} &=&\alpha _{j\cdot }g_{t}\mathbf{1}_{r\times p_{2}}+\mathbf{1}%
_{1\times l}h_{t}\beta ^{\prime }+\widetilde{R}_{j\cdot }\widetilde{F}_{t}%
\widetilde{C}^{\prime }+E_{j\cdot ,t},  \label{ex1-row} \\
X_{\cdot i,t} &=&\alpha g_{t}\mathbf{1}_{r\times 1}+\mathbf{1}_{p_{1}\times
l}h_{t}\beta _{\cdot i}^{\prime }+\widetilde{R}\widetilde{F}_{t}\widetilde{C}%
_{\cdot i}^{\prime }+E_{\cdot i,t}.  \label{ex1-column}
\end{eqnarray}%
Equations (\ref{ex1-row}) and (\ref{ex1-column}) lend themselves to the
following interpretation. }The term $g_{t}$ represents the common global
factors affecting all countries - but the rows of $\alpha $ are
heterogeneous, indicating that, for each country, the data have specific
loadings on the global factors. Similarly, $h_{t}$ are latent common factors
reflecting economic states across macroeconomic indicators - but the rows
of $\beta $ are heterogeneous, indicating that each macroeconomic indicator
loads on the states differently. Considering the rows $X_{j\cdot ,t}$ and
model (\ref{ex1-row}), the term $\alpha _{j\cdot }g_{t}\mathbf{1}_{r\times
p_{2}}$ is therefore a global factor term, while the term $\mathbf{1}%
_{1\times l}h_{t}\beta ^{\prime }$ is a columnwise adjusted global factor
term. Looking at the columns $X_{\cdot i,t}$, i.e. (\ref{ex1-column}), the
term $\mathbf{1}_{p_{1}\times l}h_{t}\beta _{\cdot i}^{\prime }$ contains
economic state factors invariant across all indicators, while the term $%
\alpha g_{t}\mathbf{1}_{r\times 1}$ is a rowwise adjusted economic state
factor term. The third terms in both (\ref{ex1-row}) and (\ref{ex1-column})
reflect the interaction effect between the two cross-sections. Hence, the
matrix factor model incorporates \textsc{simultaneously} the geographical
global factors common to the countries, and the economic state factors
common to the indicators.

Equation (\ref{ex1}) nests several interesting special cases. Indeed, when $%
r=l=1\ $and $g_{t}=h_{t}=1$, (\ref{ex1}) boils down to
\begin{equation}
X_{t}=\alpha _{p_{1}\times 1}\mathbf{1}_{1\times p_{2}}+\mathbf{1}%
_{p_{1}\times 1}\beta _{1\times p_{2}}+\widetilde{R}\widetilde{F}_{t}%
\widetilde{C}^{\prime }+E_{t},  \label{ex2}
\end{equation}%
which is a model with time-invariant, fixed effects along both the row and
column dimensions.\footnote{%
See \citet{kong2022matrix}, where this model is studied.} Such fixed effects
are allowed to be heterogeneous across the rows and/or columns, representing
the specific effects (factors) of rows (countries) and columns (indicators);
using the notation $\alpha _{p_{1}\times 1}=(r_{1},...,r_{p_{1}})^{\prime }$
and $\beta _{p_{2}\times 1}=(c_{1},...,c_{p_{2}})^{\prime }$, (\ref{ex2})
entails that%
\begin{equation}
X_{ji,t}=r_{j}+c_{i}+\widetilde{R}_{j\cdot }\widetilde{F}_{t}\widetilde{C}%
_{\cdot i}^{\prime }+E_{ji,t},  \label{ex1-scalar}
\end{equation}%
i.e., a model with \textquotedblleft two-way\textquotedblright\
cross-sectional fixed effects. Finally, combining (\ref{ex1}) and (\ref{ex2}%
) yields another special example of the matrix factor model
\begin{equation}
X_{t}=\underbrace{\alpha g_{t}\mathbf{1}_{r\times p_{2}}+\alpha
_{p_{1}\times 1}\mathbf{1}_{1\times p_{2}}}_{I}+\underbrace{\mathbf{1}%
_{p_{1}\times l}h_{t}\beta ^{\prime }+\mathbf{1}_{p_{1}\times 1}\beta
_{1\times p_{2}}}_{II}+\widetilde{R}\widetilde{F}_{t}\widetilde{C}^{\prime
}+E_{t},  \label{ex3}
\end{equation}%
i.e. a model with: fixed effects in both the row and column dimensions, two
sets of latent factors common to countries and indicators respectively, and
an interaction term.

As far as inference is concerned, under (\ref{kron}) the \textquotedblleft
loadings\textquotedblright\ $C\otimes R$ can be estimated by obtaining\ $%
\widehat{C}$ and $\widehat{R}$ separately, and subsequently computing\ $%
\widehat{C}\otimes \widehat{R}$. As pointed out in \citet{fan2021}, if one
were to estimate $C\otimes R$ by ignoring the Kronecker product structure
and using e.g. the standard PCA\ estimator studied in \citet{bai03}, the
convergence rate of $\widehat{C\otimes R}$ in $L_{2}$-norm would be $\min
\left\{ T^{-1/2},\left( p_{1}p_{2}\right) ^{-1/2}\right\} $. Conversely,
under (\ref{kron}), \citet{hkyz2021} show that the $L_{2}$-norm convergence
rates of $\widehat{C}$ and $\widehat{R}$ are, respectively%
\begin{equation*}
\min \left\{ \frac{1}{\sqrt{Tp_{1}}},\frac{1}{\min \left\{ p_{1},T\right\}
p_{2}}\right\} \text{ \ and \ }\min \left\{ \frac{1}{\sqrt{Tp_{2}}},\frac{1}{%
\min \left\{ p_{2},T\right\} p_{1}}\right\} .
\end{equation*}%
Hence, $\widehat{C}\otimes \widehat{R}$ has a faster rate of convergence
than $\widehat{C\otimes R}$ in the case of large dimensional datasets.
Further, as far as second order properties are concerned, an estimation
technique that makes full use of the dimensionality reduction implied by (%
\ref{kron}) is bound to result in efficiency gains. Finally, if the object
of interest are $C$ and $R$, direct estimation is going to be better (as
well as computationally more efficient) than firstly estimating $C\otimes R$
and subsequently recovering $C$ and $R$ therefrom via Kronecker product
decomposition (\citealp{cai2019}).

\bigskip

As mentioned above, the first contribution to consider a factor model with a
Kronecker product structure like (\ref{kron}) is \citet{wang2019factor}, who
propose estimators of the factor loading matrices (and of the numbers of the
row and column factors) based on the eigen-analysis of the
auto-cross-covariance matrix. From a different perspective, and assuming
cross-sectional pervasiveness along the row and column dimensions, %
\citet{fan2021} propose an estimation technique based on the eigen-analysis
of a weighted average of the mean and the column (row) covariance matrix of
the data; \citet{hkyz2021} improve the estimation efficiency of the factor
loading matrices with iterative projection algorithms. All these
methodologies can also be employed to construct estimators of the number of
common factors. In addition, there are also contributions which specifically
address the estimation of the dimensions of the factor spaces. In the
broader context of tensor factor models, \citet{han2020rank} propose two
approaches (one which is similar, in spirit, to the information criteria in %
\citet{baing02}, and one which is based on using the ratio of consecutive
eigenvalues) to determine the dimension of the factor spaces; %
\citet{lam2021rank} considers estimating the number of common factors by
thresholding the eigenvalues of the correlation matrix of the data (see also %
\citealp{chenlam}). Further extensions and applications of the basic set-up
in (\ref{fm}) include the constrained version by \citet{chen2019constrained}%
, the semiparametric estimators by \citet{chen2020semiparametric}, and the
estimators developed in \citet{chen2021factor}; see also \citet{han2020rank}%
. \citet{Chen2020Modeling} apply (\ref{fm}) to the dynamic transport network
in the context of international trade flows, and \citet{chen2020testing}
consider applications to financial datasets.

However, even though the literature has produced several contributions to
carry out inference in (\ref{fm}), no works has been done so far to
seriously test the existence of the factor structure implicitly defined in (%
\ref{fm}). Being able to discern whether a genuine matrix factor structure
exists or not is a crucial point in the analysis of matrix-valued data. As %
\citet{fan2021} put it, \textquotedblleft \lbrack ...] analyzing large scale
matrix-variate data is still in its infancy, and as a result, scientists
frequently analyze matrix-variate observations by separately modeling each
dimension or `flattening' them into vectors. This destroys the intrinsic
multi-dimensional structure and misses important patterns in such large
scale data with complex structures, and thus leads to sub-optimal
results\textquotedblright .

\subsection{Hypotheses of interest and the contribution of this paper\label%
{contrib}}

In (\ref{fm}), both $k_{1}$ and $k_{2}$ are strictly positive, thus allowing
for a collaborative dependence between the row cross-section and the column
cross-section: we name this \textit{two-way factor structure}. Since we
interpret $k_{1}$ and $k_{2}$ as the numbers of row and column factors, we
let $k_{1}=0$ and $k_{2}=0$ correspond to the scenarios without row factors
and without column factors, respectively. When $k_{2}=0$ but $k_{1}>0$, we
refer to this as having a \textit{one-way factor structure} along the row
dimension: all columns of the whole matrix sequence could be modeled by a $%
p_{1}$ dimensional vector factor model with effective sample size $Tp_{2}$.
A similar interpretation applies to the scenario where $k_{1}=0$ and $%
k_{2}>0 $. Finally, when $k_{1}=k_{2}=0$, the matrix-valued data is simply a
noise matrix.

In order to model the \textquotedblleft boundary\textquotedblright\ cases
discussed above, henceforth, we use the following \textsc{convention}
\begin{equation}
\underset{p_{1}\times p_{2}}{X_{t}}=\left\{
\begin{array}{ll}
\underset{p_{1}\times k_{1}}{R}\underset{k_{1}\times p_{2}}{F_{t}}+\underset{%
p_{1}\times p_{2}}{E_{t}}, & k_{1}>0,k_{2}=0, \\
\underset{p_{1}\times k_{2}}{F_{t}}\underset{k_{2}\times p_{2}}{C^{\prime }}+%
\underset{p_{1}\times p_{2}}{E_{t}}, & k_{2}>0,k_{1}=0, \\
\underset{p_{1}\times p_{2}}{E_{t}}, & k_{1}=k_{2}=0,%
\end{array}%
\right.  \label{fm1}
\end{equation}%
where the first case refers to a one-way factor model along the row
dimension (all columns form a vector factor model), the second case is a
one-way factor model along the column dimension (all rows form a vector
factor model), and the third case means absence of any factor structure. We
note that - since a factor structure is well-defined only if the dimension
of the factor space is finite - in (\ref{fm1}) we prefer to avoid the
notation $k_{2}=p_{2}\ $and$\ C=I_{p_{2}}$ (resp. $k_{1}=p_{1}\ $and$\
R=I_{p_{1}}$), even though it is mathematically equivalent to the first
(resp. the second) case in (\ref{fm1}).

\bigskip

In the context of (\ref{fm}) and (\ref{fm1}), several questions naturally
arise: \textit{is there a common, latent factor structure in the rows and/or
columns of }$X_{t}$\textit{?} \textit{How many row and/or column factors are
there?} Considering the macroeconomic example discussed above, this entails
checking the existence of country and/or index factors, and determining
their numbers. In this contribution, we propose a test to verify whether a
(one-way or two-way) matrix factor structure exists or not. To the best of
our knowledge, this is the first work with a hypothesis testing procedure to
discern between a genuine two-way matrix factor model (i.e. (\ref{fm})), a
one-way matrix factor structure (i.e. the first two cases of (\ref{fm1})),
or no factors at all (i.e. the last case of (\ref{fm1})). Our procedures
serve as a model checking tool to draw practical implications, e.g. on the
estimation technique to be employed.

Formally, we develop tests for the following general hypotheses:
\begin{equation}
H_{i0}:k_{i}\geq k_{i}^{0},\ \ \text{v. s.}\ \ H_{i1}:k_{i}<k_{i}^{0},\ \
i=1,2,  \label{Q1}
\end{equation}%
where $k_{1}^{0}$ and $k_{2}^{0}$ are the hypothesized numbers of row and
column factors, respectively. Our tests exploit the eigen-gap property of
the second moment matrix of the matrix series: we show that if there are $%
k_{i}^{0}$ common row (or column) factors, then the largest $k_{i}^{0}$
eigenvalues diverge \textsc{almost surely}, as the matrix dimensions
increase, at a faster rate than the remaining ones. To the best of our
knowledge, for the first time in the literature of matrix factor analysis,
this paper obtains an almost-sure (not just in probability) diverging lower
bound of the largest $k_{i}^{0}$ eigenvalues of the column (or row)
covariance matrix with and without projection, and an almost-sure upper
bound of the remaining eigenvalues. We then exploit the almost-sure
eigen-gap, thereby constructing a randomised test in a similar manner to %
\citet{trapani2018randomized}. In order to avoid the non-reproducibility
issue of randomised tests, we propose a \textquotedblleft
strong\textquotedblright\ rule to decide between $H_{i0}$ and $H_{i1}$,
inspired by the Law of the Iterated Logarithm.

Our approach has several desirable features. First, it is based on testing,
and therefore it does not suffer from the arbitrariness in thresholding the
eigenvalues, which is typical of information criteria. Second, it can also
be used to test for $H_{i0}:k_{i}\geq 1$ versus $H_{i1}:k_{i}=0$, thus
avoiding the arbitrariness of having to create an \textquotedblleft
artificial\textquotedblright\ eigenvalue, which is typically used to
initialise procedures based on eigenvalue ratios. Third, our tests - and
therefore our decision rules - do not require any restrictions on the
relative rates of divergence of $p_{1},p_{2}\ $and$\ T$ as they pass to
infinity, nor do they require the white noise assumption on the
idiosyncratic error matrix as in \citet{wang2019factor}. As far as the last
point is concerned, we would like to mention that the set-up by %
\citet{wang2019factor} (see also \citealp{LY12}) assumes that $E_{t}$ is
white noise, although, as a trade-off, less restrictive assumptions are
needed on the cross-sectional correlation among the components of $E_{t}$.
In the context of such a set-up, the factor model can be validated by using
existing high-dimensional white noise tests. Conversely, in the context of
an approximate factor model like ours, the issue of model validation has not
been fully investigated, i.e. no test exists to check that there is indeed a
factor structure. Our paper fills the gap in literature, and, in general, is
applicable to a wide variety of datasets.

In addition to diagnosing matrix structures, tests for (\ref{Q1}) can be
cast in a sequential procedure, as e.g. in \citet{onatski09} and %
\citet{trapani2018randomized}, thereby obtaining an estimator for the number
of common row (and/or column) factors. To the best of our knowledge, this is
the first estimator of the numbers of row and/or column factors specifically
designed for large matrix sequence, not based on eigenvalue thresholding.
After determining the common factor dimensions, it is possible to apply the
inferential theory developed e.g. in \citet{fan2021}, or \citet{hkyz2021}.
We propose two methodologies to test for (\ref{Q1}), based on the
eigenvalues of two different sample second moment matrices. Our first
procedure is based on evaluating the $k_{i}^{0}$-th largest eigenvalues of
the row (when $i=1$) and column (when $i=2$) \textquotedblleft
flattened\textquotedblright\ sample covariance matrices, defined as%
\begin{eqnarray*}
M_{c} &:&=\frac{1}{Tp_{2}}\sum_{t=1}^{T}X_{t}X_{t}^{\prime }=\frac{1}{Tp_{2}}%
\sum_{t=1}^{T}\sum_{i=1}^{p_{2}}X_{\cdot i,t}X_{\cdot i,t}^{\prime }, \\
M_{r} &:&=\frac{1}{Tp_{1}}\sum_{t=1}^{T}X_{t}^{\prime }X_{t}=\frac{1}{Tp_{1}}%
\sum_{t=1}^{T}\sum_{j=1}^{p_{1}}X_{j\cdot ,t}X_{j\cdot ,t}^{\prime },
\end{eqnarray*}%
where $X_{\cdot i,t}$\ denotes the $i$-th column of $X_{t}$, and $X_{j\cdot
,t}$\ its $j$-th row. This testing procedure is computationally
straightforward, and it requires only one step. On the other hand, using $%
M_{c}$ and $M_{r}$ ignores the two-way factor structure in model (\ref{fm}).
Hence, we also propose a second, two-step methodology which makes full use
of the low-rank structure of the common component matrix in (\ref{fm}). In
particular, we test for (\ref{Q1}) based on the column covariance matrix of
a projected matrix time series, inspired by \citet{hkyz2021}.

We would like to point out that our set-up, despite its generality, still
requires some restrictions on the data generating process of $X_{t}$.
Indeed, whilst we allow for weak cross-sectional dependence among the
idiosyncratic components, we would like to point out that the recent
contribution by \citet{lam2021rank} considers a different, stronger form of
dependence in the idiosyncratic{\ errors}, arising from the presence of weak
common factors. Further, in our theory, we do not consider the presence of
weak factors (see, however, the discussion in Section \ref{weak-factors}),
which may be viewed as a shortcoming of our set-up; however, in Section \ref%
{weak-factors}, we briefly discuss this case, indicating that it can also be
studied with our methodology. Also, our estimator of, say, $k_{1}$ based on
the projection estimator of \citet{hkyz2021} requires $k_{2}>0$, which
therefore must be tested beforehand (see the discussion in Section \ref%
{k2wrong}). Finally, a key requirement for our approach is that the
specification in (\ref{fm}){\ and (\ref{fm1})} is correct, i.e. that there
is a Kronecker product structure in the loadings as indicated in (\ref{kron}%
); in the concluding section, we further discuss the implications of this
assumption.

\bigskip

The rest of the paper is organized as follows. Section \ref{spectra}
presents the main assumptions and results on the spectra of $M_{c}$ and $%
M_{r}$, as well as the projection-based second moment matrices. Section \ref%
{inference} gives two hypotheses testing procedures for (\ref{Q1}), and the
sequential testing methodology to determine $k_{i}$ for $i=1$ and $2$; in
particular, our \textquotedblleft strong\textquotedblright\ rule to decide
between $H_{i0}$ and $H_{i1}$ is given in Section \ref{strong}. We evaluate
our theory through an extensive simulation exercise in Section \ref%
{simulation}, and we further illustrate our findings through two empirical
applications in Section \ref{empirical}. Section \ref{conclusion} concludes
the paper and discusses some avenues for future research.

To end this section, we introduce some further notation in addition to the
one already defined above. Positive finite constants are denoted as $c_{0}$,
$c_{1}$, ..., and their values may change from line to line. Throughout the
paper, we use the short-hand notation \textquotedblleft
a.s.\textquotedblright\ for \textquotedblleft almost
sure(ly)\textquotedblright . Given two sequences $a_{p_{1},p_{2},T}$ and $%
b_{p_{1},p_{2},T}$, we say that $a_{p_{1},p_{2},T}=o_{a.s.}\left(
b_{p_{1},p_{2},T}\right) $ if, as $\min \left\{ p_{1},p_{2},T\right\}
\rightarrow \infty $, it holds that $a_{p_{1},p_{2},T}b_{p_{1},p_{2},T}^{-1}%
\rightarrow 0$ a.s.; we say that $a_{p_{1},p_{2},T}=O_{a.s.}\left(
b_{p_{1},p_{2},T}\right) $ to denote that as $\min \left\{
p_{1},p_{2},T\right\} \rightarrow \infty $, it holds that $%
a_{p_{1},p_{2},T}b_{p_{1},p_{2},T}^{-1}\rightarrow c_{0}<\infty $ a.s.; and
we use the notation $a_{p_{1},p_{2},T}=\Omega _{a.s.}\left(
b_{p_{1},p_{2},T}\right) $ to indicate that as $\min \left\{
p_{1},p_{2},T\right\} \rightarrow \infty $, it holds that $%
a_{p_{1},p_{2},T}b_{p_{1},p_{2},T}^{-1}\rightarrow c_{0}>0$ a.s. Given an $%
m\times n$ matrix $A$, we denote its transpose as $A^{\prime }$ and its
element in position $\left\{ i,j\right\} $ as $A_{ij}$ or $a_{ij}$, i.e.
using either upper or lower case letters. Further, we denote the spectral
norm as $\left\Vert A\right\Vert $; we use $\left\Vert A\right\Vert _{\max }$
to denote the maximum of the absolute values of $A$'s elements; finally, we
let $\lambda _{i}\left( A\right) $ be the $i$-th largest eigenvalue of $A$.
Other, relevant notation is introduced later on in the paper.

\section{Spectra\label{spectra}}

We study the eigenvalues of the covariance matrices $M_{c}$ and $M_{r}$, and
of the projected versions (denoted as $\widetilde{M}_{c}$ and $\widetilde{M}%
_{r}$). In both cases, we find that the matrices have an eigen-gap between
the first $k_{1}$ (resp. $k_{2}$) eigenvalues and the remaining ones. As the
cross-sectional sample size $p_{1}$ (resp. $p_{2}$), increases, the first $%
k_{1}$ (resp. $k_{2}$) eigenvalues diverge at a faster rate than the
remaining ones.

\subsection{Assumptions\label{assumptions}}

The following assumptions are borrowed from the paper by \cite{hkyz2021}, to
which we refer for detailed explanations. In Section \ref{discu-ass} of the
Supplementary Material, we discuss some of our assumptions in greater detail.

\begin{assumpB}
\label{factors} (i) (a) $E(F_{t})=0$, and (b) $E\Vert F_{t}\Vert
^{4+\epsilon }\leq c_{0}$, for some $\epsilon >0$; (ii) when $k_{i}>0$ for $%
i=1,2$, it holds that
\begin{equation}
\frac{1}{T}\sum_{t=1}^{T}F_{t}F_{t}^{\prime }\overset{a.s.}{\rightarrow }{%
\Sigma }_{1}\text{ and }\frac{1}{T}\sum_{t=1}^{T}F_{t}^{\prime }F_{t}\overset%
{a.s.}{\rightarrow }\Sigma _{2},  \label{equ:covariance}
\end{equation}%
where $\Sigma _{i}$ is a $k_{i}\times k_{i}$ positive definite matrix with
distinct eigenvalues, $\lambda _{\max }\left( \Sigma _{i}\right) <\infty $,
and spectral decomposition $\Sigma _{i}=\Gamma _{i}\Lambda _{i}\Gamma
_{i}^{\prime }$. The factor numbers $k_{1}$ and $k_{2}$ are fixed as $\min
\{T,p_{1},p_{2}\}\rightarrow \infty $; (iii) it holds that, for all $%
h_{1},l_{1}$ and $h_{2},l_{2}$%
\begin{equation*}
E\max_{1\leq \widetilde{t}\leq T}\left( \sum_{t=1}^{\widetilde{t}}\left(
F_{h_{1}h_{2},t}F_{l_{1}l_{2},t}-E\left(
F_{h_{1}h_{2},t}F_{l_{1}l_{2},t}\right) \right) \right) ^{2}\leq c_{0}T;
\end{equation*}%
(iv) (a) when $k_{2}=0$ and $k_{1}>0$, it holds that
\begin{equation*}
\lambda _{\max }\left( \frac{1}{T}\sum_{t=1}^{T}F_{t}^{\prime }F_{t}\right)
=O_{a.s.}\left( \left( 1+\sqrt{\frac{p_{2}}{T}}\right) ^{2}\right) \text{ \
and \ }\frac{1}{Tp_{2}}\sum_{t=1}^{T}F_{t}F_{t}^{\prime }\overset{a.s.}{%
\rightarrow }{\Sigma }_{1}^{\ast }\text{,}
\end{equation*}%
with ${\Sigma }_{1}^{\ast }$ a $k_{1}\times k_{1}$ positive definite matrix
with distinct eigenvalues and $\lambda _{\max }\left( \Sigma _{1}^{\ast
}\right) <\infty $; (b) when $k_{1}=0$ and $k_{2}>0$, it holds that
\begin{equation*}
\lambda _{\max }\left( \frac{1}{T}\sum_{t=1}^{T}F_{t}F_{t}^{\prime }\right)
=O_{a.s.}\left( \left( 1+\sqrt{\frac{p_{1}}{T}}\right) ^{2}\right) \text{ \
and \ }\frac{1}{Tp_{1}}\sum_{t=1}^{T}F_{t}^{\prime }F_{t}\overset{a.s.}{%
\rightarrow }\Sigma _{2}^{\ast },
\end{equation*}%
with ${\Sigma }_{2}^{\ast }$ a $k_{2}\times k_{2}$ positive definite matrix
with distinct eigenvalues and $\lambda _{\max }\left( \Sigma _{2}^{\ast
}\right) <\infty $.
\end{assumpB}

\begin{assumpB}
\label{loading} (i) $\Vert R\Vert_{\max}\leq c_{0}$, and $\Vert
C\Vert_{\max} \leq c_{1}$; (ii) as $\min \{p_{1},p_{2}\}\rightarrow \infty $%
, $\Vert p_{1}^{-1}R^{\prime }R-I_{k_{1}}\Vert \rightarrow 0$ and $\Vert
p_{2}^{-1}C^{\prime }C-I_{k_{2}}\Vert \rightarrow 0$.
\end{assumpB}

Assumptions \ref{factors} and \ref{loading} are standard in large factor
models, and we refer, for example, to \citet{fan2021}. In Assumption \ref%
{factors}\textit{(i)}(b), note the (mild) strengthening of the customarily
assumed fourth moment existence condition on $F_{t}$ - this is required in
order to prove our results, which rely on almost sure rates. Similarly, the
maximal inequality in part \textit{(iii)} of the assumption is usually not
considered in the literature, and it can be derived from more primitive
dependence assumptions: for example, it can be shown to hold under various
mixing conditions (see e.g. \citealp{rio1995maximal}; and %
\citealp{shao1995maximal}); in Section \ref{bernou} in the Supplementary
Material, we show its validity for the very general class of decomposable
Bernoulli shifts (see e.g. \citealp{wu2005}). Part \textit{(iv)} of the
assumption is needed to study the case where $k_{i}=0$ - in that case,
according to (\ref{fm1}), $F_{t}$ is \textquotedblleft
large\textquotedblright\ along one dimension. The bound on $\lambda _{\max
}\left( T^{-1}\sum_{t=1}^{T}F_{t}^{\prime }F_{t}\right) $ is a high-level
condition, which we borrow from the literature on large Random Matrix Theory
(RMT; see the seminal paper by \citealp{geman}, and the review in %
\citealp{elkaroui}). In Section \ref{alt-asy} in the Supplementary Material,
we also discuss what happens under more primitive assumptions which do not
require the use of RMT.

Finally, we point out that, according to Assumption \ref{loading}, the
common factors are pervasive. Extensions to the case of \textquotedblleft
weak\textquotedblright\ factors - where the norms of $R$ and $C$ diverge at
a slower rate than $p_{1}^{1/2}$ and $p_{2}^{1/2}$ - are briefly discussed
in Section \ref{weak-factors}.

\begin{assumpB}
\label{idiosyncratic} (i) (a) $E(e_{ij,t})=0$, and (b) $E\left\vert
e_{ij,t}\right\vert ^{8}\leq c_{0}$; (ii) for all $1\leq t\leq T$, $1\leq
i\leq p_{1}$ and $1\leq j\leq p_{2}$,
\begin{equation*}
(\text{a}).\sum_{s=1}^{T}\sum_{l=1}^{p_{1}}%
\sum_{h=1}^{p_{2}}|E(e_{ij,t}e_{lh,s})|\leq c_{0},\quad (\text{b}%
).\sum_{l=1}^{p_{1}}\sum_{h=1}^{p_{2}}|{E}(e_{lj,t}e_{ih,t})|\leq c_{0};
\end{equation*}%
(iii) for all $1\leq t\leq T$, $1\leq i,l_{1}\leq p_{1}$ and $1\leq
j,h_{1}\leq p_{2}$,%
\begin{equation*}
\begin{array}{cl}
\left( \text{a}\right) . & \sum_{s=1}^{T}\sum_{l_{2}=1}^{p_{1}}%
\sum_{h=1}^{p_{2}}\left\vert
Cov(e_{ij,t}e_{l_{1}j,t},e_{ih,s}e_{l_{2}h,s})\right\vert \leq c_{0}, \\
& \sum_{s=1}^{T}\sum_{l=1}^{p_{1}}\sum_{h_{2}=1}^{p_{2}}\left\vert
Cov(e_{ij,t}e_{ih_{1},t},e_{lj,s}e_{lh_{2},s})\right\vert \leq c_{0}, \\
& \sum_{s=1}^{T}\sum_{l=1}^{p_{1}}\sum_{h=1}^{p_{2}}\left\vert
Cov(e_{ij,t}^{2},e_{lh,s}^{2})\right\vert \leq c_{0}, \\
\left( \text{b}\right) . & \sum_{s=1}^{T}\sum_{l_{2}=1}^{p_{1}}%
\sum_{h_{2}=1}^{p_{2}}\left\vert
Cov(e_{ij,t}e_{l_{1}h_{1},t},e_{ij,s}e_{l_{2}h_{2},s})+Cov(e_{l_{1}j,t}e_{ih_{1},t},e_{l_{2}j,s}e_{ih_{2},s})\right\vert \leq c_{0},%
\end{array}%
\end{equation*}%
(iv) it holds that $\lambda _{\min }\left[ E\left( \frac{1}{p_{2}T}%
\sum_{t=1}^{T}E_{t}E_{t}^{\prime }\right) \right] >0$ and $\lambda _{\min }%
\left[ E\left( \frac{1}{p_{1}T}\sum_{t=1}^{T}E_{t}^{\prime }E_{t}\right) %
\right] >0$.
\end{assumpB}

Assumption \ref{idiosyncratic} ensures the (cross-sectional and time series)
summability of the idiosyncratic terms $E_{t}$. The assumption allows for
(weak) dependence in both the space and time domains, and - as also
mentioned in the introduction - it can be read in conjunction with the paper
by \citet{wang2019factor}, where $E_{t}$ is assumed to be white noise, but
no structure is assumed on its covariance matrix. In Section \ref{bernou} in
the Supplementary Material, we show that the time-series properties of $%
E_{t} $ (in particular parts \textit{(ii)} and \textit{(iii)}, which are
high-level assumptions) are satisfied, similarly to Assumption \ref{factors}%
, by the wide class of decomposable Bernoulli shifts.

\begin{assumpB}
\label{depFE} (i) For any deterministic vectors ${v}$ and ${w}$ satisfying $%
\Vert {v}\Vert =1$ and $\Vert {w}\Vert =1$ with suitable dimensions,
\begin{equation*}
{E}\bigg\|\frac{1}{\sqrt{T}}\sum_{t=1}^{T}F_{t}({v}^{\prime }E_{t}{w})\bigg\|%
^{2}\leq c_{0};
\end{equation*}%
(ii) for all $1\leq i,l_{1}\leq p_{1}$ and $1\leq j,h_{1}\leq p_{2}$,{\small
\begin{equation*}
\begin{split}
& (\text{a}).\Big\|\sum_{h=1}^{p_{2}}{E}(\bar{\zeta}_{ij}\otimes \bar{\zeta}%
_{ih})\Big\|_{\max }\leq c_{0},\quad \Big\|\sum_{l=1}^{p_{1}}{E}(\bar{\zeta}%
_{ij}\otimes \bar{\zeta}_{lj})\Big\|_{\max }\leq c_{0}, \\
& (\text{b}).\Big\|\sum_{l=1}^{p_{1}}\sum_{h_{2}=1}^{p_{2}}Cov(\bar{\zeta}%
_{ij}\otimes \bar{\zeta}_{ih_{1}},\bar{\zeta}_{lj}\otimes \bar{\zeta}%
_{lh_{2}})\Big\|_{\max }\leq c_{0},\Big\|\sum_{l_{2}=1}^{p_{1}}%
\sum_{h=1}^{p_{2}}Cov(\bar{\zeta}_{ij}\otimes \bar{\zeta}_{l_{1}j},\bar{\zeta%
}_{ih}\otimes \bar{\zeta}_{l_{2}h})\Big\|_{\max }\leq c_{0},
\end{split}%
\end{equation*}%
}where $\bar{\zeta}_{ij}=\text{Vec}(\sum_{t=1}^{T}F_{t}e_{ij,t}/\sqrt{T})$;
(iii) (a) when $k_{2}=0$, it holds that
\begin{equation*}
\max_{1\leq h,h^{\prime }\leq
k_{2}}\sum_{j=1}^{p_{2}}\sum_{t=1}^{T}\left\vert E\left(
F_{ih,t}F_{i^{\prime }h^{\prime },s}e_{ij,t}e_{i^{\prime }j^{\prime
},s}\right) \right\vert \leq c_{0},
\end{equation*}%
for all $j^{\prime }\neq j$ and $s\neq t$; (b) when $k_{1}=0$, it holds that
\begin{equation*}
\max_{1\leq h,h^{\prime }\leq
k_{1}}\sum_{j=1}^{p_{1}}\sum_{t=1}^{T}\left\vert E\left(
F_{hi,t}F_{h^{\prime }i^{\prime },s}e_{ij,t}e_{i^{\prime }j^{\prime
},s}\right) \right\vert \leq c_{0},
\end{equation*}%
for all $j^{\prime }\neq j$ and $s\neq t$.
\end{assumpB}

According to Assumption \ref{depFE}, the common factors $F_{t}$ and the
errors $E_{t}$ can be weakly correlated. Part \textit{(i)} of the assumption
is similar to e.g. Assumption D in \citet{bai03}, in the context of vector
factor models, and it is easy to see that it is satisfied e.g. when $%
\{F_{t}\}$ and $\{E_{t}\}$ are two mutually independent groups. As far as
part \textit{(ii)} is concerned, this is a more high-level assumption which
is required in order for Lemma B.3 in \citet{hkyz2021} to hold; in turn,
this ensures that the \textquotedblleft initial\textquotedblright\
estimators of $R$ and $C$ required in Section \ref{HeKongLu}\ are
consistent, also providing a rate for them. The assumption is similar, in
spirit, to Assumption 4\textit{(3)} in \citet{fan2021}, and to Assumption D
in \citet{bai03}, in the case of vector valued series. In Section \ref%
{measurable} in the Supplementary Material, we discuss some cases in which
this part of Assumption \ref{depFE} is satisfied, including the case where $%
e_{ij,t}=g\left( F_{t}\right) u_{ij,t}$, with $u_{ij,t}$ independent across $%
i$, $j$, and $t$, and $g\left( \cdot \right) $ a measurable function; again,
a similar case is also mentioned in the discussion of Assumption D in %
\citet{bai03}.
\begin{comment}
Recall that a vector process $\{z_{t},-\infty <t<\infty \}$ is $\alpha $%
-mixing with mixing numbers $\alpha^z(h)$ $=$ $\sup_{t}\sup_{A\in \mathcal{F}%
_{-\infty }^{t},B\in \mathcal{F}_{t+h}^{\infty }}|P(A\cap B)-P(A)P(B)|$,
where $\mathcal{F}_{\tau }^{s}$ is the $\sigma $-field generated by $%
\{z_{t}:\tau \leq t\leq s\}$ if $\alpha (h)\rightarrow 0$ as $h\rightarrow
\infty $.

\begin{assumpB}
\label{mixing} (i) $\{vec(F_{t}),1\leq t\leq T\}$ is an $\alpha $-mixing
process with mixing numbers satisfying%
\begin{equation*}
\sum_{h=1}^{\infty }\left( \alpha^{F}(h)\right) ^{1-2/r}<\infty ,
\end{equation*}%
where $r$ is defined in Assumption \ref{factors}(ii); (ii) $%
\{vec(E_{t}),1\leq t\leq T\}$ is an $\alpha $-mixing process with mixing
numbers satisfying%
\begin{equation*}
\sum_{h=1}^{\infty }\left( \alpha^{E}(h)\right) ^{1-2/r^{\prime }}<\infty ,
\end{equation*}%
for some $r^{\prime }>2$.
\end{assumpB}

Assumption \ref{mixing} allows for serial dependence in the common factors
and in the error terms, and, consequently, in the whole process $X_{t}$. We
note that existing methodologies to determine the number of factors in
matrix factor models usually require the assumption that $E_{t}$ is white
noise, which may not be always satisfied.
\end{comment}

\subsection{The spectra of $M_{c}$ and $M_{r}$\label{trapani}}

To avoid repetitions, we only present results for $M_{c}$; the spectrum of $%
M_{r}$ can be studied exactly in the same way. We use the short-hand
notation $\lambda _{j}$ to indicate the $j$-th largest eigenvalue of the
expectation of $M_{c}$, and use $\widehat{\lambda }_{j}$ denote the $j$-th
largest eigenvalue of $M_{c}$.

Our first theorem provides an a.s. eigen-gap for $M_{c}$.

\begin{theorem}
\label{theorem:tildeMc} Suppose that Assumptions \ref{factors}-\ref{depFE}
are satisfied. When $k_{1}>0$, it holds that
\begin{equation}
\widehat{\lambda }_{j}=\Omega _{a.s.}\left( p_{1}\right) ,  \label{trap1}
\end{equation}%
for all $j\leq k_{1}$; further, there exist a constant $c_{0}<\infty $ such
that
\begin{equation}
\widehat{\lambda }_{j}=c_{0}+o_{a.s.}\left( \frac{p_{1}}{\sqrt{Tp_{2}}}%
\left( \ln ^{2}p_{1}\ln p_{2}\ln T\right) ^{1/2+\epsilon }\right) ,
\label{trap2}
\end{equation}%
for all $j>k_{1}$, and all $\epsilon >0$. When $k_{1}=0$, it holds that%
\begin{equation}
\widehat{\lambda }_{j}=c_{0}+O_{a.s.}\left( \frac{p_{1}}{T}\right)
+o_{a.s.}\left( \frac{p_{1}}{\sqrt{Tp_{2}}}\left( \ln ^{2}p_{1}\ln p_{2}\ln
T\right) ^{1/2+\epsilon }\right) ,  \label{trap3}
\end{equation}%
for all $j\geq 1$, and all $\epsilon >0$.
\end{theorem}

The eigen-gap in the spectrum of $M_{c}$ is the building block to construct
a procedure to decide between $H_{i0}$ and $H_{i1}$ in (\ref{Q1}). We point
out that, although the results in (\ref{trap1}) and (\ref{trap2}) are
similar, in spirit, to the ones derived by \citet{trapani2018randomized},
here we follow a quite different method of proof. Using the approach in %
\citet{trapani2018randomized}, we would be able to show only the rate $%
o_{a.s.}\left( T^{-1/2}p_{1}\left( \ln ^{2}p_{1}\ln p_{2}\ln T\right)
^{1/2+\epsilon }\right) $ in (\ref{trap2}), thus having a (much) worse rate;
moreover, the case $k_{i}=0$, where $F_{t}$ has growing dimension, is not
covered by \citet{trapani2018randomized}. As far as (\ref{trap3}) is
concerned, we note that this is a consequence of Assumption \ref{factors}%
\textit{(iv)}, and in particular of the bound $\lambda _{\max }\left(
T^{-1}\sum_{t=1}^{T}F_{t}F_{t}^{\prime }\right) =O_{a.s.}\left( \left( 1+%
\sqrt{p_{1}/T}\right) ^{2}\right) $ - see also Section \ref{alt-asy} in the
Supplementary Material for a discussion.

\subsection{The spectra of projected covariance matrices\label{HeKongLu}}

The matrices $M_{c}$ and $M_{r}$ are straightforward to compute and use, but
they are based on the implicit assumption that a \textquotedblleft
one-way\textquotedblright\ factor structure is present only in the columns
(or rows) of the observations.

When $k_{2}>0$, we propose to fully make use of the two-way interactive
factor structure in (\ref{fm}), by studying the spectrum of a projected
column (row) covariance matrix, as suggested by \citet{hkyz2021}.
Heuristically, if $k_{2}>0$ and if $C$ is known and satisfies the
orthogonality condition $C^{\prime }C/p_{2}=I_{k_{2}}$, the data matrix can
be projected into a lower dimensional space by setting $Y_{t}=X_{t}C/p_{2}$.
In view of this, we define
\begin{equation*}
\widetilde{M}_{c}=\frac{1}{T}\sum_{t=1}^{T}\widetilde{Y}_{t}\widetilde{Y}%
_{t}^{\prime },
\end{equation*}%
where $\widetilde{Y}_{t}=X_{t}\widehat{C}{/{p_{2}}}$ and $\widehat{C}$ is an
initial estimator of $C$ ($\widetilde{M}_{r}$ can be defined similarly). As
suggested by \citet{hkyz2021}, the initial estimator can be set as $\widehat{%
C}=\sqrt{p_{2}}Q$, where the columns of $Q$ are the leading $k_{2}$
eigenvectors of $M_{r}$.

\bigskip

Let $\widetilde{\lambda }_{j}$ denote the $j$-th largest eigenvalue of $%
\widetilde{M}_{c}$. The following result measures the eigen-gap of $%
\widetilde{M}_{c}$.

\begin{theorem}
\label{theorem:tildeM1} We assume that Assumptions \ref{factors}-\ref{depFE}
are satisfied and that $k_{2}>0$. When $k_{1}>0$, it holds that
\begin{equation}
\widetilde{\lambda }_{j}=\Omega _{a.s.}\left( p_{1}\right) ,
\label{tildelarge}
\end{equation}%
for all $j\leq k_{1}$; further, it holds that
\begin{equation}
\widetilde{\lambda }_{j}=o_{a.s.}\left( \left( \frac{1}{p_{2}}+\frac{1}{T}+%
\frac{p_{1}}{\sqrt{Tp_{2}}}\right) \left( \ln ^{2}p_{1}\ln p_{2}\ln T\right)
^{1+\epsilon }\right) ,  \label{tildesmall}
\end{equation}%
for all $j>k_{1}$ and all $\epsilon >0$. When $k_{1}=0$, it holds that%
\begin{equation}
\widetilde{\lambda }_{j}=O_{a.s.}\left( \frac{p_{1}}{T}\right)
+o_{a.s.}\left( \left( \frac{1}{p_{2}}+\frac{1}{T}+\frac{p_{1}}{\sqrt{Tp_{2}}%
}\right) \left( \ln ^{2}p_{1}\ln p_{2}\ln T\right) ^{1+\epsilon }\right) ,
\label{trap4}
\end{equation}%
for all $j\geq 1$, and all $\epsilon >0$.
\end{theorem}

Comparing (\ref{tildesmall}) with (\ref{trap2}) in Theorem \ref%
{theorem:tildeMc}, the eigen-gap of $\widetilde{M}_{c}$ is wider than that
of $M_{c}$. Thus, using $\widetilde{M}_{c}$ should yield a higher testing
power and a better estimate of $k_{1}$ (and/or $k_{2}$) if the two-way
interactive factor structure is really true in practice. Of course, this is
predicated upon having $k_{2}>0$. As also mentioned after Theorem \ref%
{theorem:tildeMc}, the rate in (\ref{tildesmall}) is sharper than one would
find following method of proof in \citet{trapani2018randomized}; even in
this case, we would only obtain the rate $o_{a.s.}\left( T^{-1/2}p_{1}\left(
\ln ^{2}p_{1}\ln p_{2}\ln T\right) ^{1+\epsilon }\right) $, which again
would be sub-optimal. Finally, the case where $k_{1}=0$ and $k_{2}>0$ is
covered by equation (\ref{trap4}): the same comments as for (\ref{trap3})
apply in this case (see also Section \ref{alt-asy}).

\section{Inference on the number of factors\label{inference}}

In this section, we investigate two related problems, based on determining
the dimension of the (row or column) factor structures. For brevity, we only
report results concerning $k_{1}$, but all our procedures can be readily
extended to analyse $k_{2}$.

\bigskip

We begin by presenting the tests for the null that $H_{0}:k_{1}\geq
k_{1}^{0} $ for a given $k_{1}^{0}$ (we omit the subscript $i$ in $H_{i0}$
for simplicity). We then apply these to determining whether there is a
factor structure; if this is the case, we develop a sequential procedure to
determine the dimension of each factor space. Both procedures are based on
constructing, as a first step, a test based on the rates of divergence of
the eigenvalues of either $M_{c}$ or $\widetilde{M}_{c}$ (see Section \ref%
{sec:Hypothesis of interest}); and, as a second step, a decision rule to
choose between $H_{0}$ and $H_{1}$ which is not affected by the randomness
added by the researcher (Section \ref{strong}).

\subsection{Hypothesis testing and the randomised tests\label{sec:Hypothesis
of interest}}

We consider tests for
\begin{equation}
H_{0}:k_{1}\geq k_{1}^{0}\ \ \text{vs.}\ \ H_{1}:k_{1}<k_{1}^{0}\ \ \text{%
for some}\ \ k_{1}^{0}\in \lbrack 1,\ldots ,k_{\max }],
\label{equ:hypothesisnum}
\end{equation}%
where $k_{\max }$ is a pre-specified upper bound. The hypothesis in (\ref%
{equ:hypothesisnum}) is equivalent to the following hypothesis on the
eigenvalue ${\lambda }_{k_{1}^{0}}$, that's %(the same holds when using $%
%\widetilde{\lambda }_{k_{1}^{0}}$)
\begin{equation}
H_{0}:{\lambda }_{k_{1}^{0}}\geq c_{0}p_{1}\ \ \text{vs.}\ \ H_{1}:{\lambda }%
_{k_{1}^{0}}\leq c_{0}.  \label{equ:hypothesiseigen}
\end{equation}

We propose two types of test statistics for the hypothesis testing problem
in (\ref{equ:hypothesiseigen}). Let $\beta ={\ln p_{1}}/{\ln }\left( {p_{2}T}%
\right) $, and let $\delta =\delta \left( \beta \right) \in \left(
0,1\right) $, such that%
\begin{equation}
\left\{
\begin{array}{lll}
\delta =\varepsilon & \ \text{if}\ \beta \leq 1/2 &  \\
\delta =1-1/(2\beta )+\varepsilon & \ \text{if}\ \beta >1/2 &
\end{array}%
\right. ,  \label{equ:deltabeta}
\end{equation}%
where $\varepsilon >0$ is an arbitrarily small, user-defined number. Given $%
\delta $, we define
\begin{equation}
\widehat{\phi }_{k_{1}^{0}}=\exp \left\{ \frac{p_{1}^{-\delta }\widehat{%
\lambda }_{k_{1}^{0}}}{p_{1}^{-1}\sum_{j=1}^{p_{1}}\widehat{\lambda }_{j}}%
\right\} -1\ \text{and}\ \widetilde{\phi }_{k_{1}^{0}}=\exp \left\{ \frac{%
p_{1}^{-\delta }\widetilde{\lambda }_{k_{1}^{0}}}{p_{1}^{-1}%
\sum_{j=1}^{p_{1}}\widetilde{\lambda }_{j}}\right\} -1;  \label{stats}
\end{equation}%
these are transformations of $\widehat{\lambda }_{k_{1}^{0}}$ and $%
\widetilde{\lambda }_{k_{1}^{0}}$, rescaled by the trace of $M_{c}$ and $%
\widetilde{M}_{c}$ respectively, to make them scale-invariant.

The choice of $\delta $\ in (\ref{equ:deltabeta}) is an important
specification. Its purpose is the same as in \citet{trapani2018randomized},
i.e. to make $p_{1}^{-\delta }\widehat{\lambda }_{k_{1}^{0}}$\ (and $%
p_{1}^{-\delta }\widetilde{\lambda }_{k_{1}^{0}}$) drift to zero when $%
\lambda _{k_{1}^{0}}\leq c_{0}$. In this case, it is easy to see that
rescaling by $p_{1}^{-\delta }$\ gets rid of the estimation error, while
still allowing $p_{1}^{-\delta }\widehat{\lambda }_{k_{1}^{0}}$\ (and $%
p_{1}^{-\delta }\widetilde{\lambda }_{k_{1}^{0}}$) to pass to infinity if $%
\lambda _{k_{1}^{0}}$\ does diverge. It can be verified that the value of $%
\delta $\ in (\ref{equ:deltabeta}) suffices to make the estimation error
drift to zero in all cases covered by Theorems \ref{theorem:tildeMc} and \ref%
{theorem:tildeM1}, with the exception of the cases covered by (\ref{trap3})
and (\ref{trap4}), i.e. when testing for $H_{0}:k_{1}\geq 1$, versus$\
H_{1}:k_{1}=0$. In that case, based on the rates in (\ref{trap3}) and (\ref%
{trap4}), it can be shown that (\ref{equ:deltabeta}) can be employed after
replacing $\beta $ with $\beta ^{\prime }=\ln p_{1}/\min \left\{ {\ln }%
\left( {p_{2}T}\right) ,2\ln \left( T\right) \right\} $.

\bigskip

We now turn to discussing how to use $\widehat{\phi }_{k_{1}^{0}}$\ and $%
\widetilde{\phi }_{k_{1}^{0}}$\ to test for $H_{0}:k_{1}\geq k_{1}^{0}$\ in (%
\ref{equ:hypothesiseigen}). Theorems \ref{theorem:tildeMc} and \ref%
{theorem:tildeM1} provide rates for $\widehat{\lambda }_{k_{1}^{0}}$\ and $%
\widetilde{\lambda }_{k_{1}^{0}}$\ (and, consequently, also for $\widehat{%
\phi }_{k_{1}^{0}}$\ and $\widetilde{\phi }_{k_{1}^{0}}$) under both the
null and the alternative in the hypothesis testing framework in (\ref%
{equ:hypothesiseigen}), but no limiting distribution is available. Hence, we
propose to randomise $\widehat{\phi }_{k_{1}^{0}}$\ and $\widetilde{\phi }%
_{k_{1}^{0}}$, in a similar way to \citet{trapani2018randomized}:

\begin{description}
\item[\textit{Step 1}] Generate \textit{i.i.d.} samples $\{\eta
^{(m)}\}_{m=1}^{M}$ with common distribution $N(0,1)$.

\item[\textit{Step 2}] Given $\left\{ \eta ^{(m)}\right\} _{m=1}^{M}$,
construct sample sets $\{\widehat{\psi }_{k_{1}^{0}}^{(m)}(u)\}_{m=1}^{M}$
and $\{\widetilde{\psi }_{k_{1}^{0}}^{(m)}(u)\}_{m=1}^{M}$ as
\begin{equation*}
\widehat{\psi }_{k_{1}^{0}}^{(m)}(u)=I\left[ \sqrt{\widehat{\phi }%
_{k_{1}^{0}}}\times \eta ^{(m)}\leq u\right] ,\ \ \widetilde{\psi }%
_{k_{1}^{0}}^{(m)}(u)=I\left[ \sqrt{\widetilde{\phi }_{k_{1}^{0}}}\times
\eta ^{(m)}\leq u\right] .
\end{equation*}

\item[\textit{Step 3}] Define%
\begin{equation}
\widehat{\nu }_{k_{1}^{0}}(u)=\frac{2}{\sqrt{M}}\sum_{m=1}^{M}\left[
\widehat{\psi }_{k_{1}^{0}}^{(m)}(u)-\frac{1}{2}\right] ,\ \ \widetilde{\nu }%
_{k_{1}^{0}}(u)=\frac{2}{\sqrt{M}}\sum_{m=1}^{M}\left[ \widetilde{\psi }%
_{k_{1}^{0}}^{(m)}(u)-\frac{1}{2}\right] .  \label{equ:clt}
\end{equation}

\item[\textit{Step 4}] The test statistics are finally defined as
\begin{equation*}
\widehat{\Psi }_{k_{1}^{0}}=\int_{U}\left[ \widehat{\nu }_{k_{1}^{0}}(u)%
\right] ^{2}dF(u),\ \ \widetilde{\Psi }_{k_{1}^{0}}=\int_{U}\left[
\widetilde{\nu }_{k_{1}^{0}}(u)\right] ^{2}dF(u),
\end{equation*}%
where $F(u)$ is a weight function.
\end{description}

The test described above is similar to the one proposed in %
\citet{trapani2018randomized}; however, in the construction of $\widehat{%
\Psi }_{k_{1}^{0}}$ and $\widetilde{\Psi }_{k_{1}^{0}}$, we propose a
weighted average across different values of $u$ through the weight function $%
F\left( u\right) $. As a consequence, it can be expected that the test will
not be affected by an individual value of $u$, a form of scale invariance
which is not considered in \citet{trapani2018randomized}.

\begin{assumpC}
\label{weight} $F\left( u\right) $ is a differentiable function for all $%
u\in U$ such that: (i) $\int_{U}dF\left( u\right) =1$; (ii) $%
\int_{U}u^{2}dF\left( u\right) <\infty $.
\end{assumpC}

Assumption \ref{weight} is satisfied by several functions, the most
\textquotedblleft natural\textquotedblright\ candidates being distribution
functions; we discuss in detail $F\left( u\right) $, its possible
specifications, and how to compute integrals involving it in Section \ref%
{wfunction} of the Supplementary Material.

\bigskip

Let $P^{\ast }$ denote the probability law of $\{\widehat{\psi }%
_{k_{1}^{0}}^{(m)}(u)\}_{m=1}^{M}$ and $\{\widetilde{\psi }%
_{k_{1}^{0}}^{(m)}(u)\}_{m=1}^{M}$ conditional on the sample $\{X_{t},1\leq
t\leq T\}$, and \textquotedblleft $\overset{D^{\ast }}{\rightarrow }$%
\textquotedblright\ and \textquotedblleft $\overset{P^{\ast }}{\rightarrow }$%
\textquotedblright\ as convergence in distribution and in probability,
respectively, according to $P^{\ast }$.

\begin{proposition}
\label{theorem:nulldis} We assume that Assumptions {\ref{factors}-\ref{depFE}
}and \ref{weight} are satisfied. Then, under $H_{0}:k_{1}\geq k_{1}^{0}$, as
$\min \{p_{1},p_{2},T,M\}\rightarrow \infty $ with
\begin{equation}
M\exp \left\{ -\epsilon p_{1}^{1-\delta }\right\} \rightarrow 0,
\label{restrictM}
\end{equation}%
for some $0<\epsilon <c_{0}/\bar{\lambda}$ and $\bar{\lambda}%
=p_{1}^{-1}\sum_{j=1}^{p_{1}}\lambda _{j}$, it holds that
\begin{equation}
\widehat{\Psi }_{k_{1}^{0}}\overset{D^{\ast }}{\rightarrow }\chi _{1}^{2},
\label{asy-null}
\end{equation}%
for almost all realisations of $\{X_{t},1\leq t\leq T\}$. Under the same
assumptions, if $k_{2}>0$ it also holds that $\widetilde{\Psi }_{k_{1}^{0}}%
\overset{D^{\ast }}{\rightarrow }\chi _{1}^{2}$ for almost all realisations
of $\{X_{t},1\leq t\leq T\}$.

Under $H_{1}:\lambda _{k_{1}^{0}}\leq c_{0}<\infty $, as $\min
\{p_{1},p_{2},T,M\}\rightarrow \infty $ it holds that%
\begin{equation}
M^{-1}\widehat{\Psi }_{k_{1}^{0}}\overset{P^{\ast }}{\rightarrow }c_{1},
\label{asy-alt}
\end{equation}%
for some $0<c_{1}<\infty $ and almost all realisations of $\{X_{t},1\leq
t\leq T\}$. Under the same assumptions, if $k_{2}>0$ it also holds that $%
M^{-1}\widetilde{\Psi }_{k_{1}^{0}}\overset{P^{\ast }}{\rightarrow }c_{1}$,
for almost all realisations of $\{X_{t},1\leq t\leq T\}$.
\end{proposition}

Equation (\ref{asy-null}) states that, under the null, both test statistics $%
\widehat{\Psi }_{k_{1}^{0}}$ and $\widetilde{\Psi }_{k_{1}^{0}}$ converge in
distribution to a chi-square distribution with one degree of freedom. This
can be understood heuristically by noting that, under the null, both $%
\widehat{\phi }_{k_{1}^{0}}$ and $\widetilde{\phi }_{k_{1}^{0}}$ go to
infinity, and therefore the variances of $\sqrt{\widehat{\phi }_{k_{1}^{0}}}%
\times \eta ^{(m)}$ and $\sqrt{\widetilde{\phi }_{k_{1}^{0}}}\times \eta
^{(m)}$ also pass to infinity. Thus, heuristically, $\{\widehat{\psi }%
_{k_{1}^{0}}^{(m)}(u)\}_{m=1}^{M}$ and $\{\widetilde{\psi }%
_{k_{1}^{0}}^{(m)}(u)\}_{m=1}^{M}$ follow a Bernoulli distribution with
success probability $1/2$. By the Central Limit Theorem, in (\ref{equ:clt})
as $M$ goes to infinity, both $\widehat{\nu }_{k_{1}^{0}}(u)$ and $%
\widetilde{\nu }_{k_{1}^{0}}(u)$ follow the standard normal distribution $%
N(0,1)$ (conditional on the sample) asymptotically. The results hold for all
samples, save for a zero measure set. By Proposition \ref{theorem:nulldis},
it follows immediately that, for almost all realisations of $\{X_{t},1\leq
t\leq T\}$%
\begin{equation}
\lim_{\min \{p_{1},p_{2},T,M\}\rightarrow \infty }P^{\ast }\left( \widehat{%
\Psi }_{k_{1}^{0}}>c_{\alpha }|H_{0}\right) =\alpha ,  \label{size}
\end{equation}%
where $c_{\alpha }$ is such that $P\left( \chi _{1}^{2}>c_{\alpha }\right)
=\alpha $, and
\begin{equation}
\lim_{\min \{p_{1},p_{2},T,M\}\rightarrow \infty }P^{\ast }\left( \widehat{%
\Psi }_{k_{1}^{0}}>c_{\alpha }|H_{1}\right) =1.  \label{power}
\end{equation}%
The results also hold if one substitutes $\widehat{\Psi }_{k_{1}^{0}}$ with $%
\widetilde{\Psi }_{k_{1}^{0}}$ in (\ref{size}) and (\ref{power}).

\subsection{A \textquotedblleft strong\textquotedblright\ rule to decide
between $H_{0}$ and $H_{1}$\label{strong}}

The tests are constructed by using added randomness, $\{\eta
^{(m)}\}_{m=1}^{M}$, whose effect does not vanish asymptotically as would be
the case e.g. when using the bootstrap. In turn, this entails that the
properties of tests based on $\widehat{\Psi }_{k_{1}^{0}}$ (and $\widetilde{%
\Psi }_{k_{1}^{0}}$) are different from the properties of \textquotedblleft
standard\textquotedblright\ tests. Indeed, equation (\ref{power}) has the
classical interpretation: whenever a researcher uses $\widehat{\Psi }%
_{k_{1}^{0}}$ (and $\widetilde{\Psi }_{k_{1}^{0}}$), (s)he will reject the
null, when false, with probability one. Conversely, the implications of (\ref%
{size}) are subtler. Due to the artificial randomness $\{\eta
^{(m)}\}_{m=1}^{M}$, different researchers using the same data will obtain
different values of $\widehat{\Psi }_{k_{1}^{0}}$ and $\widetilde{\Psi }%
_{k_{1}^{0}}$, and, consequently, different $p$-values; indeed, if an
infinite number of researchers were to carry out the test, the $p$-values
would follow a uniform distribution on $\left[ 0,1\right] $. %
\citet{corradi2006} provide an alternative explanation, writing that
\textquotedblleft \lbrack ...] as the sample size gets larger, all
researchers always reject the null when false, while $\alpha \%$ of the
researchers always reject the null when it is true\textquotedblright.
%\begin{comment}This implies that $\hat{\Psi}_{k_1^0}$ and $\tilde{\Psi}_{k_1^0}$ would falsely reject the null at a prevalent rate even if $\hat{\phi}_{k_1^0}$ and $\tilde{\phi}_{k_1^0}$ are in theory perfectly separated almost surely as the sample size increases to infinity.\end{comment}

In order to address this problem, we propose a further step which, in
essence, \textquotedblleft de-randomizes\textquotedblright\ $\widehat{\Psi }%
_{k_{1}^{0}}$ and $\widetilde{\Psi }_{k_{1}^{0}}$. Each researcher, instead
of computing $\widehat{\Psi }_{k_{1}^{0}}$ or $\widetilde{\Psi }_{k_{1}^{0}}$
just once, will compute the test statistic $S$ times, at each iteration $s$
generating a statistic $\widehat{\Psi }_{k_{1}^{0},s}$ (or $\widetilde{\Psi }%
_{k_{1}^{0},s}$) using a random sequence $\left\{ \eta _{s}^{\left( m\right)
},1\leq m\leq M\right\} $, independent across $1\leq s\leq S$, and thence
defining, for some $\alpha \in \left( 0,1\right) $%
\begin{equation}
\widehat{Q}_{k_{1}^{0}}\left( \alpha \right) =S^{-1}\sum_{s=1}^{S}I\left[
\widehat{\Psi }_{k_{1}^{0},s}\leq c_{\alpha }\right] ,  \label{q}
\end{equation}%
and the same when using $\widetilde{\Psi }_{k_{1}^{0},s}$ - in this case
obtaining $\widetilde{Q}_{k_{1}^{0}}\left( \alpha \right) $. A consequence
of Proposition \ref{theorem:nulldis} is
\begin{equation}
\begin{tabular}{ll}
$\lim_{\min (p_{1},p_{2},T,M,S)\rightarrow \infty }P^{\ast }\{\widehat{Q}%
_{k_{1}^{0}}\left( \alpha \right) =1-\alpha \}=1$ & $\text{for }%
H_{0}:k_{1}\geq k_{1}^{0},$ \\
$\lim_{\min p_{1},p_{2},T,M,S)\rightarrow \infty }P^{\ast }\{\widehat{Q}%
_{k_{1}^{0}}\left( \alpha \right) =0\}=1$ & $\text{for }%
H_{1}:k_{1}<k_{1}^{0}.$%
\end{tabular}
\label{test-function}
\end{equation}%
Equation (\ref{test-function}) stipulates that, as $S\rightarrow \infty $,
averaging across $s$ in (\ref{q}) washes out the added randomness in $%
\widehat{Q}_{k_{1}^{0}}\left( \alpha \right) $: all researchers using this
procedure will obtain the same value of $\widehat{Q}_{k_{1}^{0}}\left(
\alpha \right) $, thereby ensuring reproducibility. The function $\widehat{Q}%
_{k_{1}^{0}}\left( \alpha \right) $ corresponds to (the complement to one
of) the \textquotedblleft fuzzy decision\textquotedblright , or
\textquotedblleft abstract randomised decision rule\textquotedblright\
reported in equation (1.1a) in \citet{geyer}. \citet{geyer} (see also %
\citealp{geyer2005rejoinder}) provide a helpful discussion of the meaning of
$\widehat{Q}_{k_{1}^{0}}\left( \alpha \right) $: the problem of deciding in
favour or against $H_{0}$ may be modelled through a random variable, say $D$%
, which can take two values, namely \textquotedblleft do not reject $H_{0}$%
\textquotedblright\ and \textquotedblleft reject $H_{0}$\textquotedblright .
Such a random variable has probability $\widehat{Q}_{k_{1}^{0}}\left( \alpha
\right) $ to take the value \textquotedblleft do not reject $H_{0}$%
\textquotedblright , and probability $1-\widehat{Q}_{k_{1}^{0}}\left( \alpha
\right) $ to take the value \textquotedblleft reject $H_{0}$%
\textquotedblright . In this context, (\ref{test-function}) states that
(asymptotically), the probability of the event $\left\{ \omega :D=\text{%
\textquotedblleft reject }H_{0}\text{\textquotedblright }\right\} $ is $%
\alpha $ when $H_{0}$ is satisfied, for all researchers - corresponding to
the notion of \textit{size} of a test; see also the quote from %
\citet{corradi2006} reported above. Conversely, under $H_{1}$, the
probability of the event $\left\{ \omega :D=\text{\textquotedblleft reject }%
H_{0}\text{\textquotedblright }\right\} $ is $1$ (asymptotically),
corresponding to the notion of \textit{power}.

\bigskip

Reporting the value of $\widehat{Q}_{k_{1}^{0}}\left( \alpha \right) $ or $%
\widetilde{Q}_{k_{1}^{0}}\left( \alpha \right) $ could be sufficient in some
applications. In our case, the individual tests for $H_{0}:k_{1}\geq
k_{1}^{0}$ will form the basis of a sequential procedure to provide an
estimate of $k_{1}$, and therefore we also need a decision rule to choose,
based on $\widehat{Q}_{k_{1}^{0}}\left( \alpha \right) $ (or $\widetilde{Q}%
_{k_{1}^{0}}\left( \alpha \right) $), between $H_{0}$ and $H_{1}$. We base
such a decision rule on a Law of the Iterated Logarithm for $\widehat{Q}%
_{k_{1}^{0}}\left( \alpha \right) $ and $\widetilde{Q}_{k_{1}^{0}}\left(
\alpha \right) $.

\begin{theorem}
\label{strong-rule} We assume that Assumptions {\ref{factors}-\ref{depFE}}
and \ref{weight} are satisfied, and that $M=O\left( T\right) $ and $S=\Omega
\left( M\right) $. Then it holds that%
\begin{equation}
\frac{\widehat{Q}_{k_{1}^{0}}\left( \alpha \right) -\left( 1-\alpha \right)
}{\sqrt{\alpha \left( 1-\alpha \right) }}=\Omega _{a.s.}\left( \sqrt{\frac{%
2\ln \ln S}{S}}\right) ,\ \frac{\widetilde{Q}_{k_{1}^{0}}\left( \alpha
\right) -\left( 1-\alpha \right) }{\sqrt{\alpha \left( 1-\alpha \right) }}%
=\Omega _{a.s.}\left( \sqrt{\frac{2\ln \ln S}{S}}\right) ,  \label{st1}
\end{equation}%
under $H_{0}:k_{1}\geq k_{1}^{0}$, for almost all realisations of $%
\{X_{t},1\leq t\leq T\}$. Also, it holds that%
\begin{equation}
\widehat{Q}_{k_{1}^{0}}\left( \alpha \right) =o_{a.s.}\left( 1\right) \ \
\text{and}\ \ \widetilde{Q}_{k_{1}^{0}}\left( \alpha \right) =o_{a.s.}\left(
1\right) ,  \label{st2}
\end{equation}%
under $H_{1}:k_{1}<k_{1}^{0}$, for almost all realisations of $\{X_{t},1\leq
t\leq T\}$.
\end{theorem}

Equations (\ref{st1}) and (\ref{st2}) complement (\ref{test-function}), and
quantify the gap in the asymptotic behaviour of $\widehat{Q}%
_{k_{1}^{0}}\left( \alpha \right) $ (or $\widetilde{Q}_{k_{1}^{0}}\left(
\alpha \right) $) according as the null $H_{0}$, or the alternative $H_{1}$,
is satisfied. According to the theorem, $\widehat{Q}_{k_{1}^{0}}\left(
\alpha \right) $\ drifts to zero under the alternative; conversely, (\ref%
{st1}) entails that, for sufficiently large $\left( p_{1},p_{2},T\right) $,%
\footnote{%
Formally, (\ref{st1}) states that there exists a triple of random variables $%
\left( p_{1,0},p_{2,0},T_{0}\right) $ such that (\ref{lil}) holds for all $%
\left( p_{1},p_{2},T\right) $ with $p_{1}\geq p_{1,0}$, $p_{2}\geq p_{2,0}$
and $T\geq T_{0}$.} $\widehat{Q}_{k_{1}^{0}}\left( \alpha \right) $\ is
bounded away from zero with lower bound%
\begin{equation}
\widehat{Q}_{k_{1}^{0}}\left( \alpha \right) \geq 1-\alpha -\sqrt{\alpha
\left( 1-\alpha \right) }\sqrt{\frac{2\ln \ln S}{S}}.  \label{lil}
\end{equation}

This gap can be exploited to construct a decision rule based on $\widehat{Q}%
_{k_{1}^{0}}\left( \alpha \right) $ (or $\widetilde{Q}_{k_{1}^{0}}\left(
\alpha \right) $), not rejecting the null when $\widehat{Q}%
_{k_{1}^{0}}\left( \alpha \right) $ (or $\widetilde{Q}_{k_{1}^{0}}\left(
\alpha \right) $) exceeds a threshold, and rejecting otherwise. In theory,
one could use the threshold defined in (\ref{st1}), but this, albeit valid
asymptotically, is likely to be overly conservative in finite samples. A
less conservative decision rule in favour of the null could be%
\begin{equation}
\widehat{Q}_{k_{1}^{0}}\left( \alpha \right) \geq \left( 1-\alpha \right)
-f\left( S\right) ,\ \ \text{or}\ \ \widetilde{Q}_{k_{1}^{0}}\left( \alpha
\right) \geq \left( 1-\alpha \right) -f\left( S\right)  \label{thresholds}
\end{equation}%
with $f\left( S\right) $ a user-specified, non-increasing function of $S$
such that
\begin{equation}
\lim_{S\rightarrow \infty }f\left( S\right) =0\text{ and }\lim
\sup_{S\rightarrow \infty }\left( f\left( S\right) \right) ^{-1}\sqrt{\frac{%
2\ln \ln S}{S}}=0.  \label{fs-function}
\end{equation}%
We call such a family of rules \textquotedblleft strong
rules\textquotedblright , since they originate from a \textquotedblleft
strong\textquotedblright\ result (the Law of the Iterated Logarithm). Whilst
we discuss possible choices of $f\left( S\right) $\ in Sections \ref%
{simulation} and \ref{empirical}, offering guidelines based on synthetic and
real data, here we note that a typical family of default choices for $%
f\left( S\right) $\ is
\begin{equation}
f\left( S\right) =S^{-q},  \label{choice-fs}
\end{equation}%
where $0<q<1/2$. On account of (\ref{thresholds}), it can be expected that
as $q$\ increases, $f\left( S\right) $\ vanishes more quickly, making the
threshold more exacting. As a consequence, tests become less and less
conservative, thus leading to a higher probability of rejection of the null
hypothesis; in turn, this results in a potential underestimation of $%
k_{1}^{0}$\ in finite samples. Conversely, lower values of $q$\ entail that $%
f\left( S\right) $\ is larger, whence more conservative tests and,
consequently, a higher probability of overstating the number of common
factors, at least in finite samples. Based on Monte Carlo evidence, our
recommended choice is based on using $q=1/4$\ in (\ref{choice-fs}).

Finally, and along similar lines as the comment above, the theorem only
requires that $M$\ and $S$\ be of the same order of magnitude as each other,
and that they are (at most) proportional to $T$. The choice of these tuning
parameters is explored in Section \ref{simulation} (where we show that
results are robust to these specifications); here, we note that a default
choice is $M=S=T$.

\subsection{Determining the number of common factors \label{sequential}}

The output of the decision rules proposed in (\ref{thresholds}) can be used
for two purposes. Firstly, it is possible to check whether $k_{1}=0$ or $%
k_{2}=0$: this entails that there exists no factor structure in the rows or
columns. Similarly, finding $k_{1}=k_{2}=0$ implies that there is no factor
structure along the row and/or column sections. As a second application of (%
\ref{thresholds}), upon finding that $k_{1}>0$ (or $k_{2}>0$), the
individual decision rules proposed above can be cast in a sequential
procedure to determine the number of common row and column factors, based on
$M_{c}$ (and $M_{r}$) and $\widetilde{M}_{c}$ (and $\widetilde{M}_{r}$)
respectively.

As mentioned above, using $\widetilde{M}_{c}$\ and $\widetilde{M}_{r}$\
should yield better results in the presence of a genuine two-way structure.
Hence, $\widetilde{M}_{c}$\ and $\widetilde{M}_{r}$ should be employed if $%
k_{2}>0$ and $k_{1}>0$ respectively.

\bigskip

The estimator of $k_{1}$ (denoted as $\widehat{k}_{1}$ when using $\widehat{%
\Psi }_{1}$ and $\widehat{Q}_{k_{1}^{0}}\left( \alpha \right) $, and $%
\widetilde{k}_{1}$ when using $\widetilde{\Psi }_{1}$ and $\widetilde{Q}%
_{k_{1}^{0}}\left( \alpha \right) $) is the output of the following
algorithm:

\begin{description}
\item[\textit{Step 1}] Run the test for $H_{0}:k_{1}\geq 1$ based on either $%
\widehat{Q}_{1}\left( \alpha \right) $ or $\widetilde{Q}_{1}\left( \alpha
\right) $. If the null is rejected with (\ref{thresholds}), set $\widehat{k}%
_{1}=0$\ (resp. $\widetilde{k}_{1}=0$) and stop, otherwise go to the next
step.

\item[\textit{Step 2}] For $j\geq 2$, run the test for $H_{0}:k_{1}\geq j$
based on either $\widehat{Q}_{j}\left( \alpha \right) $ or $\widetilde{Q}%
_{j}\left( \alpha \right) $, constructed using an artificial sample $\left\{
\eta _{j,s}^{(m)},1\leq m\leq M\right\} $ generated independently across $%
1\leq s\leq S$, and independently of $\left\{ \eta _{1,s}^{(m)},1\leq m\leq
M\right\} ,$ $...,$ $\left\{ \eta _{j-1,s}^{(m)},1\leq m\leq M\right\} $. If
the null is rejected with (\ref{thresholds}), set $\widehat{k}_{1}=j-1$\
(resp. $\widetilde{k}_{1}=j-1$) and stop; otherwise repeat step 2 until the
null is rejected, or until a pre-specified value $k_{\max }$ is reached.
\end{description}

The consistency of $\widehat{k}_{1}$ and $\widetilde{k}_{1}$ is stated in
the next theorem.

\begin{theorem}
\label{fwise-str} We assume that the assumptions of Theorem \ref{strong-rule}
are satisfied, and that $k_{1}\leq k_{\max }$. Then, as $\min \left\{
p_{1},p_{2},T,M,S\right\} \rightarrow \infty $, it holds that $P^{\ast
}\left( \widehat{k}_{1}=k_{1}\right) =1$ for\ almost all realisations of $%
\{X_{t},1\leq t\leq T\}$. If, further, $k_{2}>0$, then it holds that $%
P^{\ast }\left( \widetilde{k}_{1}=k_{1}\right) =1$, for\ almost all
realisations of $\{X_{t},1\leq t\leq T\}$.
\end{theorem}

\subsection{Remarks\label{remarks}}

We discuss two aspects of our methodology: the estimation of the number of
row/column factors when the number of column/row factors is unknown \textit{%
a priori}, and the performance of our methodology in the presence of weak
factors.

\subsubsection{Estimation of $k_{1}$ when $k_{2}$ is unknown\label{unknown}}

The results above are based on the (implicit) assumption that, when
estimating the number of row column factors $k_{1}$, the number of column
common factors $k_{2}$\ is known (and vice versa). In practice, an estimate
of $k_{2}$\ is required prior to computing $\widehat{k}_{1}$\ or $\widetilde{%
k}_{1}$. This can be obtained using any of the available techniques, but for
argument's sake we focus on using $\widehat{k}_{2}$\ derived from using our
sequential approach based on $\widehat{\Psi }_{k_{2}^{0}}$: whilst
suboptimal (as our Monte Carlo shows), $\widehat{k}_{2}$\ is a consistent
estimator of $k_{2}$\ according to Theorem \ref{fwise-str}; further, its
implementation does not require any prior knowledge of $k_{1}$, thus being
\textquotedblleft ready to use\textquotedblright\ as an initial estimate of $%
k_{2}$.

Further, $\widetilde{M}_{c}$\ and $\widetilde{M}_{r}$ should be used only if
$k_{2}>0$ and $k_{1}>0$ respectively. If this is not known a priori, we
recommend using a two-stage approach to estimate $k_{2}$\ (or, respectively,
$k_{1}$). In the first step, the applied user should run the test for $%
H_{0}:k_{2}\geq 1$, based on $\widehat{\phi }_{k_{2}^{0}}$\ with $%
k_{2}^{0}=1 $; according to Theorem \ref{theorem:tildeM1}, this can be
applied irrespective of whether $k_{1}=0$\ or $k_{1}>0$, and the conclusions
from this test are, therefore, robust to the actual value of $k_{1}$. Upon
rejecting the null, this entails that $k_{2}^{0}=0$, and therefore a two-way
structure does not exists. Consequently, $k_{1}$\ should be determined using
$\widehat{\phi }_{k_{1}^{0}}$, based on the spectrum of the
\textquotedblleft flattened\textquotedblright\ matrix $M_{c}$. If,
conversely, the null $H_{0}:k_{2}\geq 1$\ is not rejected, then there is a
factor structure in the columns of $X_{t}$, and the applied user can use
either $\widehat{\phi }_{k_{1}^{0}}$\ or $\widetilde{\phi }_{k_{1}^{0}}$,
based on $M_{c}$\ and $\widetilde{M}_{c}$\ respectively. Indeed, when $%
k_{2}>0$, Theorems \ref{theorem:tildeMc} and \ref{theorem:tildeM1} stipulate
that both approaches are valid, although (as discussed above) it can be
expected that using the spectrum of $\widetilde{M}_{c}$\ should lead to
improvements. Indeed, upon finding that $k_{1}>0$, $k_{2}$\ can also be
re-estimated using $\widetilde{\phi }_{k_{2}^{0}}$, thus having a
(potentially) refined estimator; conversely, if $k_{1}$\ is found to be
zero, $k_{2}$\ can be estimated using $\widehat{\phi }_{k_{2}^{0}}$. As a
final remark, in Section \ref{k2wrong} in the Supplementary Material we
explore the case where $k_{2}=0$, but the applied user employs the
projection method anyway, e.g. due to an incorrect initial estimation of $%
k_{2}$, showing that results are anyway robust to this form of
mis-specification.

\bigskip

In the next result, we show that, when estimating $k_{1}$\ using $\widehat{k}%
_{2}$, the resulting estimator is consistent, and it preserves the same mode
of convergence as in Theorem \ref{fwise-str}.

\begin{corollary}
\label{fwise-corol} We assume that the assumptions of Theorem \ref{fwise-str}
are satisfied and that $\widehat{k}_{2}$\ is used as an estimator of $k_{2}$%
. Then, as $\min \left\{ p_{1},p_{2},T,M,S\right\} \rightarrow \infty $, it
holds that $P^{\ast }\left( \widehat{k}_{1}=k_{1}\right) =1$\ for\ almost
all realisations of $\{X_{t},1\leq t\leq T\}$. If, further, $\widehat{k}%
_{2}>0$, then it holds that $P^{\ast }\left( \widetilde{k}_{1}=k_{1}\right)
=1$, for\ almost all realisations of $\{X_{t},1\leq t\leq T\}$.
\end{corollary}

\subsubsection{Extensions to the case of weak factors\label{weak-factors}}

According to Assumption \ref{loading}, the row and column factors considered
in this paper are \textquotedblleft pervasive\textquotedblright\ or
\textquotedblleft strong\textquotedblright . Technically, this is due to the
fact that the (squared) $L_{2}$-norm of the loading matrices diverge at
rates $\left\Vert R^{\prime }R\right\Vert =\Omega \left( p_{1}\right) $ and $%
C^{\prime }C=\Omega \left( p_{2}\right) $. However, the literature on vector
factor models has recently considered the case of common factors where the
square of the $L_{2}$-norm of the loading matrix still diverges, but at a
rate slower than the cross-sectional dimension - in our context, this would
e.g. correspond to having $\left\Vert R^{\prime }R\right\Vert =\Omega \left(
p_{1}^{\alpha _{1}}\right) $, for some $0<\alpha _{1}<1$. The recent
contributions by \citet{uematsu2} and \citet{uematsu}, and the references
cited therein, offer a state-of-the-art discussion of the issue of
determining weak factors in the context of vector factor models. As far as
matrix factor models are concerned, the literature has also considered the
possible presence of weak factors: for example, \citet{wang2019factor}, %
\citet{chen2019constrained} and \citet{gaotsay} all consider weak factors
along both the row and column spaces, modelling their strength with an
approach similar to the one in this paper; and, as discussed in the
introduction, \citet{lam2021rank} allows for the presence of weak common
factors in the idiosyncratic error term, thereby allowing for factor-induced
cross-correlation.

Hence, in this section we briefly investigate how our methodology works in
the presence of weak factors. For the sake of a concise discussion, we focus
primarily on determining the presence of only one weak common factor ($%
k_{1}=1$) in the case where both the row and column common factors are
possibly weak.

\bigskip

In the context of factor models for matrix valued data, we assume that
Assumptions \ref{factors}-\ref{weight} above all hold, but we
replace/integrate Assumption \ref{loading} with the following

\begin{assumpB}
\label{loading-weak} Assumption \ref{loading} holds with part (ii) replaced
by $\Vert p_{1}^{-\alpha _{1}}R^{\prime }R-I_{k_{1}}\Vert \rightarrow 0$ as $%
p_{1}\rightarrow \infty $, for some $0<\alpha _{1}\leq 1$, and $\Vert
p_{2}^{-\alpha _{2}}C^{\prime }C-I_{k_{2}}\Vert \rightarrow 0$ as $%
p_{2}\rightarrow \infty $, for some $0<\alpha _{2}\leq 1$.
\end{assumpB}

This assumption summarizes the discussion at the beginning of this section:
the (squared) $L_{2}$-norm of the loading matrix $R$ diverges, at a rate
that is possibly \textit{lower} than $p_{1}$. Similarly, the column factors
can also be weak. Prior to reporting the main theoretical result, we offer a
heuristic preview of the main arguments. Repeating the proofs of our main
results, it can be shown that, for all $j\leq k_{1}$%
\begin{equation*}
\widehat{\lambda }_{j}=\Omega _{a.s.}\left( p_{1}^{\alpha _{1}}\right) \text{
\ and \ }\widetilde{\lambda }_{j}=\Omega _{a.s.}\left( p_{1}^{\alpha
_{1}}\right) ,
\end{equation*}%
whereas the conclusions of Theorems \ref{theorem:tildeMc} and \ref%
{theorem:tildeM1} still hold true for all $j>k_{1}$. All methodologies that
try to determine the number of common factors require some eigen-gap in the
second order matrices. This entails that - in the case of our methodology -
detection of weak factors is in principle possible as long as, as $\min
\left\{ p_{1},p_{2},T\right\} \rightarrow \infty $%
\begin{equation}
\frac{\widehat{\lambda }_{k_{1}}}{\widehat{\lambda }_{k_{1}+1}}\overset{a.s.}%
{\rightarrow }\infty .  \label{divergence}
\end{equation}

The following result summarises the ability of $\widehat{k}_{1}$ to estimate
the number of common factors in the presence of weak factors.

\begin{corollary}
\label{weak} We assume that the assumptions of Theorem \ref{sequential}
hold, with Assumption \ref{loading} replaced with Assumption \ref%
{loading-weak}. Then, as\textbf{\ }$\min \left\{ p_{1},p_{2},T,M,S\right\}
\rightarrow \infty $, if%
\begin{eqnarray}
p_{1}^{\alpha _{1}}p_{2}^{\alpha _{2}-1} &\rightarrow &\infty ,
\label{weak-restrictoin} \\
p_{1}^{\alpha _{1}-1}p_{2}^{\alpha _{2}-1/2}T^{1/2} &\rightarrow &\infty ,
\label{weak-restriction-1}
\end{eqnarray}%
then it holds that $P^{\ast }\left( \widehat{k}_{1}=k_{1}\right) =1$, for\
almost all realisations of $\{X_{t},1\leq t\leq T\}$.
\end{corollary}

The results in Corollary \ref{weak} can be read in conjunction with similar
results in \citet{wang2019factor}, \citet{chen2019constrained} and %
\citet{gaotsay}.\ Some comments on (\ref{weak-restrictoin}) and (\ref%
{weak-restriction-1}) are in order. A quick inspection of Theorem \ref%
{theorem:tildeM1} reveals that the leading eigenvalues of $M_{c}$ are
proportional to $\Omega \left( p_{1}^{\alpha _{1}}p_{2}^{\alpha
_{2}-1}\right) $; hence, (\ref{weak-restrictoin}) ensures that such
eigenvalues diverge, which is a necessary condition to find factors.
Equation (\ref{weak-restriction-1}) entails that (\ref{divergence}) holds,
and therefore the \textquotedblleft signal\textquotedblright\ associated to
common factors is not drown out by the estimation \textquotedblleft
noise\textquotedblright . In this respect, (\ref{weak-restriction-1})
states, heuristically, that when $p_{1}$ is \textquotedblleft too
big\textquotedblright\ (in comparison with the other dimensions, $p_{2}$ and
$T$), detection of weak factors is less easy.

Equation (\ref{weak-restriction-1}) can be illustrated through some
examples. Consider, for simplicity, $\alpha _{2}=1$ - that is, column
factors are strong. If $p_{1}=p_{2}=T$, then (\ref{weak-restriction-1})
boils down to $p_{1}^{\alpha _{1}}\rightarrow \infty $, which is always true
as long as $\alpha _{1}>0$: this entails that, in this case, arbitrarily
weak factors can be (potentially) detected. Indeed, in (\ref{stats}), $%
\widehat{\lambda }_{k_{1}^{0}}$ is dampened by a factor $p_{1}^{-\delta }$,
with - in this example - $\delta =\varepsilon $. Hence, it can be verified
that $\widehat{\phi }_{k_{1}^{0}}$ will diverge as long as $\alpha
_{1}>\varepsilon $: the choice of $\varepsilon $, which is entirely up to
the applied user, will determine which weak factors can be detected and
which ones will be left out. As another example, consider a
\textquotedblleft very long\textquotedblright\ matrix sequence, where e.g. $%
p_{1}=T^{1/2}$; in such a case, (\ref{weak-restriction-1}) becomes $%
p_{1}^{\alpha _{1}}p_{2}^{1/2}\rightarrow \infty $, which is satisfied even
when $p_{2}$ diverges very slowly (e.g., even if $p_{2}=\ln p_{1}$); in this
case, again, the choice of $\varepsilon $ will determine which weak factors
can be detected, and this is further enhanced the larger $T$ is. Conversely,
consider the case where the matrix sequence is \textquotedblleft very
short\textquotedblright , e.g. $T=p_{1}^{1/2}$. In such a case, (\ref%
{weak-restriction-1}) becomes $p_{1}^{\alpha _{1}-3/4}p_{2}^{1/2}\rightarrow
\infty $, which, if e.g. $p_{2}=p_{1}$, entails that detection is possible
only when $\alpha _{1}>1/4$: very weak factors cannot be detected in this
case. A similar phenomenon was also noted in the case of vector valued
series by \citet{trapani2018randomized}: however, in the case of matrix
valued series, a small $T$ can be offset by a large value of $p_{2}$, which
is a major advantage of having a matrix structure in the data. Other
examples can also be considered, but the general message is that detection
of weak factors is helped by both $T$ and $p_{2}$.

In the presence of common factors that are weak along the columns, i.e. when
$\alpha _{2}<1$, the interpretation of (\ref{weak-restriction-1}) is more
convoluted, but essentially the same. As mentioned above, large values of $%
p_{2}$ help the estimation of $k_{1}$: however, such helpfulness is dampened
when the common factors in the column are weak - in essence, because the
information coming from aggregating the columns is, itself, weak.

As a final remark, we note that the case of using $\widetilde{\lambda }%
_{k_{1}}$ is more complicated, essentially because the estimation error of $%
\widehat{C}$ is compounded (and inflated) by the presence of weak factors.
In the interest of brevity, we relegate the treatment of this case to Lemmas %
\ref{c-hat-weak} and \ref{weak-project1} in the Supplementary Material. The
latter result is, essentially, an equivalent of Theorem \ref{theorem:tildeMc}
in the presence of weak factors. An analogue restriction to (\ref%
{weak-restriction-1}) can be derived from Lemma \ref{weak-project1}; in
particular, equation (\ref{weak-proj}) in the lemma suggests that detection
of weak factors requires the necessary condition $p_{1}^{\alpha
_{1}}p_{2}^{\alpha _{2}-1}\rightarrow \infty $. Sufficient conditions,
similar to (\ref{weak-restriction-1}), can be derived from (\ref%
{weak-proj-22}). However, results are far more complicated, and of dubious
helpfulness. Technically, this is due to the fact that weak factors also
enter the projected estimator $\widehat{C}$, making it less precise (again,
due to the fact that estimation of $C$ is now based on a weaker
\textquotedblleft signal\textquotedblright ).

\section{Simulation studies\label{simulation}}

In this section, we evaluate the finite sample performances of our strong
rule to determine whether there is a factor structure, and of the sequential
procedure to estimate the number of common factors. As far as the latter is
concerned, we compare our Sequential Testing Procedures (henceforth denoted
as \textquotedblleft STP\textquotedblright ) with several competing
methodologies available in the literature.

\bigskip

We begin with describing the implementation of our procedures. For the
proposed STP, three different approaches can be adopted: firstly, the test
statistics can be constructed using the eigenvalues of $M_{c}$ (or $M_{r}$),
and we denote this approach as $\text{STP}_{1}$; secondly, the test
statistics can be constructed using the eigenvalues of $\widetilde{M}_{c}$
(or $\widetilde{M}_{r}$), using STP$_{1}$ as a preliminary step to estimate
e.g. $k_{2}$ and subsequently using the estimated value, $\widehat{k}_{2}$,
to construct the initial estimator $\widehat{C}$, and we denote this
approach as STP$_{3}$; and, finally, the test statistics can still be
constructed using the eigenvalues of $\widetilde{M}_{c}$ (or $\widetilde{M}%
_{r}$), but in order to avoid the (finite sample) risk of understating $%
k_{2} $ one can use a deliberately large number instead of the STP$_{1}$
estimator (we set $\widehat{k}_{2}=k_{\max }=8$), and we denote this
approach as STP$_{2}$.

When computing integrals such as $\int_{-\infty }^{\infty }\left[ \widehat{%
\nu }_{k_{1}^{0}}(u)\right] ^{2}dF(u)$, we use the distribution of the
standard normal as weight function $F\left( u\right) $, using a
Gauss-Hermite quadrature with%
\begin{equation}
\widehat{\Psi }_{k_{1}^{0}}=\sum_{s=1}^{n_{S}}w_{s}\widehat{\nu }%
_{k_{1}^{0}}(\sqrt{2}z_{s}).  \label{hermite}
\end{equation}%
In (\ref{hermite}), the $z_{s}$s, $1\leq s\leq n_{S}$, are the zeros (of the
physicist's version) of the Hermite polynomial $H_{n_{S}}\left( z\right) $
defined as%
\begin{equation}
H_{n_{S}}\left( z\right) =\left( -1\right) ^{n_{S}}\exp \left( x^{2}\right)
\frac{d^{n_{S}}}{dx^{n_{S}}}\exp \left( -x^{2}\right) ,  \label{her-poly}
\end{equation}%
and the weights $w_{s}$ are defined as
\begin{equation*}
w_{s}=\frac{2^{n_{S}-1}\left( n_{S}-1\right) !}{n_{S}\left[
H_{n_{S}-1}\left( z_{s}\right) \right] ^{2}}.
\end{equation*}%
Thus, when computing $\widehat{\nu }_{k_{1}^{0}}(u)$ in Step 2 of the
algorithm, we construct $n_{S}$ of these statistics, each using $u=\pm \sqrt{%
2}z_{s}$. The values of the roots $z_{s}$, and of the corresponding weights $%
w_{s}$, are tabulated e.g. in \citet{salzer}. In our case, we have used $%
n_{S}=4$, which corresponds to $w_{1}=w_{4}=0.05$ and $w_{2}=w_{3}=0.45$,
and $u_{1}=-u_{4}=2.4$ and $u_{2}=-u_{3}=0.7$.

\bigskip

\textit{Data generation}

\bigskip

We use the same Data Generating Process (DGP) as \citet{hkyz2021} in order
to generate $X_{t}$. Specifically, when $k_{1}\neq 0$ and $k_{2}\neq 0$,
i.e., when there is a factor structure, we generate the entries of $R$ and $%
C $ independently from the uniform distribution $\mathcal{U}(-1,1)$, and we
let%
\begin{eqnarray}
\text{Vec}\left( F_{t}\right) &=&\phi \text{Vec}\left( F_{t-1}\right) +\sqrt{%
\theta \left( 1-\phi ^{2}\right) }\text{Vec}\left( \epsilon _{t}\right) ,%
\text{ \ \ }\epsilon _{t}\sim i.i.d. \ \mathcal{N}\left( 0,I_{k_{1}\times
k_{2}}\right) ,  \label{equ41} \\
\text{Vec}\left( E_{t}\right) &=&\psi \text{Vec}\left( E_{t-1}\right) +\sqrt{%
1-\psi ^{2}}\text{Vec}\left( U_{t}\right) ,\text{ \ \ }U_{t}\sim i.i.d. \
\mathcal{N}\left( 0,V_{E}\otimes U_{E}\right) ,  \label{equ42}
\end{eqnarray}%
where $U_{E}$ and $V_{E}$ are matrices with ones on the diagonal, and the
off-diagonal entries are $a/p_{1}$ and $a/p_{2}$, respectively. The
parameter $a$ controls cross-sectional dependence, with larger $a$ leading
to stronger cross-dependence; we have used $a=2$ in our simulations. In the
case that no factor structure exists, we simply let $X_{t}=E_{t}$, with $%
E_{t}$ generated in the same way as in (\ref{equ41})-(\ref{equ42}). The
parameters $\phi $ and $\psi $ control both the temporal and cross-sectional
correlations of $X_{t}$; with nonzero $\phi $ and $\psi $, the generated
factors are temporally correlated while the idiosyncratic noises are both
temporally and cross-sectionally correlated. In all our experiments, we set $%
\phi =\psi =0.1$ and use $\theta =1$.\footnote{%
We have also tried $\phi =\psi =0.3$, and results are essentially the same.}
In all the simulation settings, the reported results are based on $500$
replications.

As a final remark, our DGP entails that, letting $\Sigma _{E}=E\left( \text{%
Vec}\left( E_{t}\right) \text{Vec}^\prime \left( E_{t}\right) \right) $, and
$\Sigma _{C}=E\Big( \text{Vec}\left( RF_{t}C^{\prime }\right)$ $\text{Vec}%
^{\prime }\left( RF_{t}C^{\prime }\right) \Big) $, the signal-to-noise ratio
is given by%
\begin{equation}
\sigma =\frac{\sum_{i,j=1}^{p_{1}p_{2}}\left\{ \Sigma _{C}\right\} _{ij}}{%
\sum_{i,j=1}^{p_{1}p_{2}}\left\{ \Sigma _{E}\right\} _{ij}}\simeq \theta
\frac{k_{1}k_{2}}{9\left( 1+a\right) ^{2}},  \label{stn}
\end{equation}%
for large values of $p_{1}$ and $p_{2}$. We point out that, in unreported
simulations we have tried to alter $\theta $; as expected, when this
increases all criteria improve (this is particularly evident for the
estimator proposed by \citealp{lam2021rank}), but the relative performance
remains unaltered.

\bigskip

\textit{Determining whether there is a factor structure}

\bigskip

We investigate the finite sample performance of our \textquotedblleft
strong\textquotedblright\ rule to determine whether a factor structure
exists in the matrix time-series data. This offers a solution to the
question in the discussion section of \citet{hkyz2021}: \textit{is the
matrix factor structure true for the time series?}

\begin{table}[h]
\caption{Proportions of correctly determining whether there exists factor
structure using $\widehat{\Psi }_{1}^{S}$ and $\widetilde{\Psi }_{1}^{S}$
over 500 replications with $M=S=300,\protect\phi=\protect\psi%
=0.1,f(S)=S^{-1/4}$. }
\label{tab:main1}\renewcommand{\arraystretch}{0.2} \centering
\selectfont
\begin{threeparttable}
   \scalebox{1}{\begin{tabular*}{16cm}{ccccccccccccccccccccccccccccc}
\toprule[2pt]
&\multirow{2}{*}{$\alpha$} &\multirow{2}{*}{$(k_1,k_2)$}  &\multirow{2}{*}{Method}  &\multicolumn{3}{c}{$(p_1,T)=(100,100)$}&\multicolumn{3}{c}{$(p_1,T)=(150,150)$} \cr
\cmidrule(lr){5-7} \cmidrule(lr){8-10}
&  &&                 &$p_2=15$     &$p_2=20$      &$p_2=30$     &$p_2=15$     &$p_2=20$      &$p_2=30$   \\
\midrule[1pt]
&0.01&(0,0)   &$\widehat{\Psi}_{1}^S$   &1 &1 &1  &1 &1 &1\\
&&       &$\widetilde{\Psi}_{1}^S$ &1 &1 &1  &1 &1 &1\\
&&       &Vec &1 &1 &1  &1 &1 &1\\
\cmidrule(lr){3-10}
&&(1,1)    &$\widehat{\Psi}_{1}^S$   &0.648 &0.838 &0.97  &0.850 &0.946 &0.998\\
&&        &$\widetilde{\Psi}_{1}^S$ &0.970 &0.998 &1.00  &0.996 &1.000 &1.000\\
&&       &Vec &0.000 &0.000 &0.00  &0.000 &0.000 &0.000\\
\cmidrule(lr){3-10}
&&(1,3)    &$\widehat{\Psi}_{1}^S$   &1 &1 &1  &1 &1 &1\\
&&        &$\widetilde{\Psi}_{1}^S$ &1 &1 &1  &1 &1 &1\\
&&       &Vec &0 &0 &0  &0 &0 &0\\
\midrule[1pt]
&0.05&(0,0)   &$\widehat{\Psi}_{1}^S$   &1 &1 &1  &1 &1 &1\\
&&       &$\widetilde{\Psi}_{1}^S$ &1 &1 &1  &1 &1 &1\\
&&       &Vec &1 &1 &1  &1 &1 &1\\
\cmidrule(lr){3-10}
&&(1,1)    &$\widehat{\Psi}_{1}^S$   &0.334 &0.574 &0.904  &0.644 &0.84 &0.984\\
&&        &$\widetilde{\Psi}_{1}^S$ &0.886 &0.992 &1.000  &0.982 &1.00 &1.000\\
&&       &Vec &0.000 &0.000 &0.000  &0.000 &0.00 &0.000\\
\cmidrule(lr){3-10}
&&(1,3)    &$\widehat{\Psi}_{1}^S$   &1 &1 &1  &1 &1 &1\\
&&        &$\widetilde{\Psi}_{1}^S$ &1 &1 &1  &1 &1 &1\\
&&       &Vec &0 &0 &0  &0 &0 &0\\
\midrule[1pt]
&0.10&(0,0)   &$\widehat{\Psi}_{1}^S$   &1 &1 &1  &1 &1 &1\\
&&       &$\widetilde{\Psi}_{1}^S$ &1 &1 &1  &1 &1 &1\\
&&       &Vec &1 &1 &1  &1 &1 &1\\
\cmidrule(lr){3-10}
&&(1,1)    &$\widehat{\Psi}_{1}^S$   &0.210 &0.422 &0.786  &0.52 &0.724 & 0.966\\
&&        &$\widetilde{\Psi}_{1}^S$ &0.812 &0.992 &1.000  &0.96 &1.000 &1.000\\
&&       &Vec &0.000 &0.000 &0.000  &0.00 &0.000 &0.000\\
\cmidrule(lr){3-10}
&&(1,3)    &$\widehat{\Psi}_{1}^S$   &1 &1 &1  &1 &1 &1\\
&&        &$\widetilde{\Psi}_{1}^S$ &1 &1 &1  &1 &1 &1\\
&&       &Vec &0 &0 &0  &0 &0 &0\\
\bottomrule[2pt]
  \end{tabular*}}
  \end{threeparttable}
\end{table}

\begin{table}[h]
\caption{Proportions of correctly determining whether there exists factor
structure using $\widehat{\Psi }_{1}^{S}$ and $\widetilde{\Psi }_{1}^{S}$
over 500 replications with $M=S=300,\protect\phi= \protect\psi%
=0.1,f(S)=S^{-1/4}$. }
\label{tab:main11}\renewcommand{\arraystretch}{0.2} \centering
\selectfont
\begin{threeparttable}
		\scalebox{0.75}{\begin{tabular*}{22.5cm}{ccccccccccccccccccccccccccccc}
				\toprule[2pt]
				&\multirow{2}{*}{$\alpha$} &\multirow{2}{*}{$(k_1,k_2)$}  &\multirow{2}{*}{Method}  &\multicolumn{3}{c}{$T=50$}&\multicolumn{3}{c}{$T=100$} \cr
				\cmidrule(lr){5-7} \cmidrule(lr){8-10}
				&  &&                 &$p_1=p_2=50$     &$p_1=p_2=100$      &$p_1=p_2=150$     &$p_1=p_2=50$     &$p_1=p_2=100$      &$p_1=p_2=150$   \\
				\midrule[1pt]
				&0.01&(0,0)   &$\widehat{\Psi}_{1}^S$   &1.000 &1 &1  &1.000 &1.000 &1.000\\
				&&       &$\widetilde{\Psi}_{1}^S$ &0.142 &1 &1  &0.766 &0.166 &0.988\\
				&&       &Vec &1.000 &1 &1  &1.000 &1.000 &1.000\\
				\cmidrule(lr){3-10}
				&&(1,1)    &$\widehat{\Psi}_{1}^S$   &0.986 &1 &1  &0.99 &1 &1\\
				&&        &$\widetilde{\Psi}_{1}^S$ &1.000 &1 &1  &1.00 &1 &1\\
				&&       &Vec &0.000 &0 &0  &0.00 &0 &0\\
				\cmidrule(lr){3-10}
				&&(1,3)    &$\widehat{\Psi}_{1}^S$   &1 &1 &1  &1 &1 &1\\
				&&        &$\widetilde{\Psi}_{1}^S$ &1 &1 &1  &1 &1 &1\\
				&&       &Vec &0 &0 &0  &0 &0 &0\\
\hline
				&0.05&(0,0)   &$\widehat{\Psi}_{1}^S$   &1.000  &1 &1  &1 &1.000 &1\\
				&&       &$\widetilde{\Psi}_{1}^S$ &0.938  &1 &1  &1 &0.996 &1\\
				&&       &Vec &1.000 &1 &1  &1 &1.000 &1\\
				\cmidrule(lr){3-10}
				&&(1,1)    &$\widehat{\Psi}_{1}^S$   &0.888 &0.998 &1  &0.936 &1 &1\\
				&&        &$\widetilde{\Psi}_{1}^S$ &1.000 &1.000 &1  &1.000 &1 &1\\
				&&       &Vec &0.000 &0.000 &0  &0.000 &0 &0\\
				\cmidrule(lr){3-10}
				&&(1,3)    &$\widehat{\Psi}_{1}^S$   &1 &1 &1  &1 &1 &1\\
				&&        &$\widetilde{\Psi}_{1}^S$ &1 &1 &1  &1 &1 &1\\
				&&       &Vec &0 &0 &0  &0 &0 &0\\
\hline
				&0.10&(0,0)   &$\widehat{\Psi}_{1}^S$   &1.000 &1 &1  &1 &1 &1\\
				&&       &$\widetilde{\Psi}_{1}^S$ &0.998 &1 &1  &1 &1 &1\\
				&&       &Vec &1.000 &1 &1  &1 &1 &1\\
				\cmidrule(lr){3-10}
				&&(1,1)    &$\widehat{\Psi}_{1}^S$   &0.794 &0.972 &1  &0.846 &1 &1\\
				&&        &$\widetilde{\Psi}_{1}^S$ &1.000 &1.000 &1  &1.000 &1 &1\\
				&&       &Vec &0.000 &0.000 &0  &0.000 &0 &0\\
				\cmidrule(lr){3-10}
				&&(1,3)    &$\widehat{\Psi}_{1}^S$   &1 &1 &1  &1 &1 &1\\
				&&        &$\widetilde{\Psi}_{1}^S$ &1 &1 &1  &1 &1 &1\\
				&&       &Vec &0 &0 &0  &0 &0 &0\\

				\bottomrule[2pt]
		\end{tabular*}}
	\end{threeparttable}
\end{table}

We study two scenarios: first, the case of \textquotedblleft
small\textquotedblright\ $p_{2}$, using $p_{1}=T=\{100,150\}$ and $%
p_{2}=\{15,20,30\}$; second, the more balanced cases $\left( p1,p2,T\right)
=\left( 50,50,50\right) $, $\left( 100,100,50\right) $, $\left(
150,150,50\right) $, $\left( 50,50,100\right) $, $\left( 100,100,100\right) $%
, and $\left( 150,150,100\right) $.\footnote{%
In Section \ref{furthermc} in the supplement, we complement these results
with two more sets of experiments, considering $p_{1}$ and $p_{2}$ being
comparable and considering $p_{1}$ and $p_{2}$ being comparable and small,
respectively. Results are broadly similar to the ones reported here.} In all
our simulations, we use $k_{\max }=8$, although we tried different values
for $k_{\max }$ and the results show that the proposed methods are not
sensitive to the choice of it. As far as our decision rule is concerned, we
have used $\alpha \in \{0.01,0.05,0.10\}$, $M=300$ and $S=300$; in (\ref%
{thresholds}), we have used $f\left( S\right) =S^{-1/4}$. Results using
different combinations of $(\alpha ,M,S)$ and different choice of $f\left(
S\right) $ are in the Supplementary Material. By way of comparison, we also
use the test developed in \citet{trapani2018randomized}, applying it to the $%
p_{1}p_{2}\times 1$ series Vec$\left( X_{t}\right) $ - this is denoted by
\textquotedblleft \textrm{Vec}\textquotedblright\ in our tables.

\bigskip

Firstly, we consider the case of the existence of a matrix factor structure
by setting $(k_{1},k_{2})=\{(1,1),(1,3)\}$ in (\ref{equ41}). We report the
proportions of correctly claiming that there exists factor structure by the
\textquotedblleft strong\textquotedblright\ rule in the second and third
rows of each subpanel of Tables \ref{tab:main1} and \ref{tab:main11}; $%
\widetilde{\Psi }_{1}^{S}$ and $\widehat{\Psi }_{1}^{S}$ denote our
procedure with and without projection technique, respectively. The results
indicate that, when using $\widetilde{\Psi }_{1}^{S}$, a factor structure is
found more than $95\%$ of the times whenever $p_{1}\geq 100$, with the sole
exception of the (small sample) case $\left( p_{1},p_{2},T\right) =\left(
100,15,100\right) $. Results are always worse when $\widehat{\Psi }_{1}^{S}$
is used, although improvements are seen for larger sample sizes (see Table %
\ref{tab:main11}): this reinforces the case in favour of the projection
method developed by \citet{hkyz2021}, especially when $p_{2}$ is small. When
$p_{1}=50$, results are essentially the same, with few exceptions. All
across the board, results obtained using the test by %
\citet{trapani2018randomized} are very bad, indicating that, in essence, the
test always fails to detect the existence of a factor structure when this is
present. This result is not entirely unexpected: based on %
\citet{trapani2018randomized}, eigenvalues are scaled by a factor $\left(
p_{1}p_{2}\right) ^{-\delta }$, where $\delta =0.685$ in the
\textquotedblleft best\textquotedblright\ case where $p_{1}=100$ and $%
p_{2}=15$. Indeed, $\delta $ increases (by construction) with $p_{2}$, thus
resulting in dampening eigenvalues even more. In turn, this makes detection
of diverging eigenvalues particularly difficult.

Secondly, we investigate the performance of our proposed methodology when
there is no factor structure in the matrix time series $X_{t}=E_{t}$, i.e. $%
(k_{1},k_{2})=(0,0)$, with $E_{t}$ generated in the same way as in (\ref%
{equ41}). In the first row of each subpanel of Tables \ref{tab:main1} and %
\ref{tab:main11}, we report the proportions of correctly claiming that there
exists no factor structure, from which we can conclude that our proposed
methodology is extremely powerful in identifying the absence of a factor
structure, even in the small sample case $(p_{1},p_{2},T)=(100,15,100)$, and
with or without projection. There are some puzzling exceptions in Table \ref%
{tab:main11}, especially when using $\alpha =0.01$, but these issues vanish
as the sample size increases. We note again that results obtained using the
test by \citet{trapani2018randomized} are, in this case, very satisfactory,
but this is clearly a spurious effect due to the reasons discussed above.

In the Supplementary Material, we report more results obtained under
different scenarios, which reinforce our conclusions - see e.g. Table \ref%
{tab:r3q71}, where we consider smaller values of $p_{1}$ and $p_{2}$. In
Tables \ref{tab:supp1} and \ref{tab:supp2}, we report results based on
different combinations of $(\alpha ,M,S)$. Results are essentially the same
when using $\widetilde{\Psi }_{1}^{S}$, whereas $\widehat{\Psi }_{1}^{S}$ is
more sensitive (at least in small samples) to the choice of $M$ and - albeit
to a lesser extent - $S$. In particular, as far as the former is concerned,
smaller values of it seem to yield better results in small samples. Finally,
in Table \ref{tab:supp3} in the Supplementary Material, we assess the
sensitivity to $f\left( S\right) $; as can be expected, results are affected
by the choice of the threshold, but this is only marginal when using $%
\widetilde{\Psi }_{1}^{S}$ and, again, more pronounced when using $\widehat{%
\Psi }_{1}^{S}$.

\bigskip

\textit{Determining the number of common factors}

\bigskip

We investigate the finite sample performances of the sequential testing
procedure introduced in Section \ref{sequential}. We use the same design,
and consider the same combinations of $\left( p_{1},p_{2},T\right) $ as in
the previous set of experiments. In order to evaluate the sequential
procedure, we use it to estimate the number of row factors $k_{1}$
considering the cases $(k_{1},k_{2})=\{(1,1),(1,3),(3,1),(3,3)\}$. As in the
previous section, we use $M=300$ and $S=300$, and $f\left( S\right)
=S^{-1/4} $, and we only report results for the case $\alpha =0.01$ for
brevity.

As well as assessing the performance of the STP methods, we compare these
against the most popular alternatives in the literature. We have considered
the following techniques:\footnote{%
Details on how each procedure has been implemented are in Section \ref%
{mc-notes} in the Supplement.} the Iterative Eigenvalue-Ratio (denoted as
\textrm{IterER}) studied in \citet{hkyz2021}; the $\alpha $-PCA
Eigenvalue-Ratio method (denote as $\alpha $\textrm{-PCA}) proposed by %
\citet{fan2021}; the iterative versions of the Eigenvalue Ratio and of the
Information Criteria algorithms by \citet{han2020rank} (here denoted as
\textrm{iTIP-ER} and \textrm{iTIP-IC} respectively); and the method proposed
by \citet{lam2021rank} (denoted as \textrm{TCorTh}). Finally, we have used
the information criterion $PC_{p1}\left( k\right) $ proposed in %
\citet{baing02}, applied to a vectorised version of $X_{t}$, to determine
the total number of factors $k=k_{1}k_{2}$ (this is denoted by \textrm{IC}).

\begin{table}[!h]
\caption{Simulation results for estimating $k_1$ in the form $x(y|z)$, $x$
is the sample mean of the estimated factor numbers based on 500 replications
$\protect\alpha =0.01,M=S=300,\protect\phi=\protect\psi=0.1,f(S)=S^{-1/4}$, $%
y$ and $z$ are the proportions of under- and exact- determination of the
factor number, respectively. }
\label{tab:main2}\renewcommand{\arraystretch}{1.5} \centering
\selectfont
\begin{threeparttable}
   \scalebox{0.65}{\begin{tabular*}{26cm}{ccccccccccccccccccccccccccccc}
\toprule[2pt]
&\multirow{2}{*}{$(k_1,k_2)$}&\multirow{2}{*}{Method}  &\multicolumn{3}{c}{$(p_1,T)=(100,100)$}&\multicolumn{3}{c}{$(p_1,T)=(150,150)$} \cr
\cmidrule(lr){4-6} \cmidrule(lr){7-9}
&&                 &$p_2=15$     &$p_2=20$       &$p_2=30$    &$p_2=15$     &$p_2=20$         &$p_2=30$ \\
\midrule[1pt]
&(1,1)   &$\text{STP}_{1}$     &$0.648 ( 0.352 | 0.648 )$	 &$0.838 ( 0.162 | 0.838 )$		 &$0.97 ( 0.03 | 0.97 )$	 &$0.85 ( 0.15 | 0.85 )$	&$0.946 ( 0.054 | 0.946 )$ &$0.998 ( 0.002 | 0.998 )$\\
&      &$\text{STP}_{2}$     &$0.968 ( 0.032 | 0.968 )$	 &$0.998 ( 0.002 | 0.998 )$		 &$1 ( 0 | 1 )$	 &$0.998 ( 0.002 | 0.998 )$	 &$1 ( 0 | 1 )$  &$1 ( 0 | 1 )$\\
&      &$\text{STP}_{3}$     &$0.706 ( 0.294 | 0.706 )$	 &$0.914 ( 0.086 | 0.914 )$		 &$0.994 ( 0.008 | 0.99 )$	 &$0.758 ( 0.242 | 0.758 )$	 &$0.942 ( 0.058 | 0.942 )$  &$0.994 ( 0.006 | 0.994 )$\\
&      &IterER               &$0.474 ( 0.526 | 0.474 )$	 &$0.766 ( 0.234 | 0.766 )$  &$0.972 ( 0.028 | 0.972 )$	 &$0.544 ( 0.456 | 0.544 )$ &$0.82 ( 0.18 | 0.82 )$  &$0.988 ( 0.012 | 0.988 )$\\
&      &$\alpha$-PCA         &$0.796 ( 0.204 | 0.796 )$	  &$0.832 ( 0.168 | 0.832 )$	 &$0.884 ( 0.116 | 0.884 )$	 &$0.956 ( 0.044 | 0.956 )$ &$0.962 ( 0.038 | 0.962 )$  &$0.986 ( 0.014 | 0.986 )$\\
&      &iTIP-IC               &$1 ( 0 | 1 )$	 &$1 ( 0 | 1 )$  &$1 ( 0 | 1 )$	 &$1 ( 0 | 1 )$ &$1 ( 0 | 1 )$  &$1 ( 0 | 1 )$\\
&      &iTIP-ER               &$1.368 ( 0 | 0.636 )$	 &$1.35 ( 0 | 0.65 )$  &$1.336 ( 0 | 0.664 )$	 &$1.39 ( 0 | 0.61 )$ &$1.328 ( 0 | 0.672 )$  &$1.352 ( 0 | 0.648 )$\\
&      &TCorTh         &$2 ( 0 | 0 )$	  &$2 ( 0 | 0 )$	 &$2 ( 0 | 0 )$	 &$2 ( 0 | 0 )$ &$2 ( 0 | 0 )$  &$2 ( 0 | 0 )$\\
&      &IC         &$1 ( 0 | 1 )$	  &$1 ( 0 | 1 )$	 &$1 ( 0 | 1 )$	 &$1 ( 0 | 1 )$ &$1 ( 0 | 1 )$  &$1 ( 0 | 1 )$\\
\cmidrule(lr){2-9}
&(1,3)   &$\text{STP}_{1}$     &$1 ( 0 | 1 )$	 &$1 ( 0 | 1 )$		 &$1 ( 0 | 1 )$	 &$1 ( 0 | 1 )$	 &$1 ( 0 | 1 )$ &$1 ( 0 | 1 )$\\
&      &$\text{STP}_{2}$     &$1 ( 0 | 1 )$	 &$1 ( 0 | 1 )$		 &$1 ( 0 | 1 )$	 &$1 ( 0 | 1 )$	 &$1 ( 0 | 1 )$  &$1 ( 0 | 1 )$\\
&      &$\text{STP}_{3}$     &$0.988 ( 0.012 | 0.988 )$	 &$1 ( 0 | 1 )$		 &$1 ( 0 | 1 )$	 &$0.996 ( 0.004 | 0.996 )$	 &$1 ( 0 | 1 )$  &$1 ( 0 | 1 )$\\
&      &IterER               &$0.97 ( 0.03 | 0.97 )$	 &$0.996 ( 0.004 | 0.996 )$  &$1 ( 0 | 1 )$	 &$0.996 ( 0.004 | 0.996 )$ &$1 ( 0 | 1 )$  &$1 ( 0 | 1 )$\\
&      &$\alpha$-PCA         &$1 ( 0 | 1 )$	  &$1 ( 0 | 1 )$	 &$1 ( 0 | 1 )$	 &$1 ( 0 | 1 )$ &$1 ( 0 | 1 )$  &$1 ( 0 | 1 )$\\
&      &iTIP-IC               &$1 ( 0 | 1 )$	 &$1 ( 0 | 1 )$  &$1 ( 0 | 1 )$	 &$1 ( 0 | 1 )$ &$1 ( 0 | 1 )$  &$1 ( 0 | 1 )$\\
&      &iTIP-ER               &$1.186 ( 0 | 0.814 )$	 &$1.174 ( 0 | 0.826 )$  &$1.144 ( 0 | 0.856 )$	 &$1.12 ( 0 | 0.88 )$ &$1.098 ( 0 | 0.902 )$  &$1.098 ( 0 | 0.902 )$\\
&      &TCorTh         &$2 ( 0 | 0 )$	  &$2 ( 0 | 0 )$	 &$2 ( 0 | 0 )$	 &$2 ( 0 | 0 )$ &$2 ( 0 | 0 )$  &$2 ( 0 | 0 )$\\
&      &IC         &$2.858 ( 0.142 | 0.858 )$	  &$2.954 ( 0.046 | 0.954 )$	 &$2.988 ( 0.012 | 0.988 )$	 &$2.968 ( 0.032 | 0.968 )$ &$2.994 ( 0.006 | 0.994 )$  &$3 ( 0 | 1 )$\\
\cmidrule(lr){2-9}
&(3,1)   &$\text{STP}_{1}$     &$2.106 ( 0.338 | 0.662 )$	 &$2.732 ( 0.12 | 0.88 )$		 &$2.994 ( 0.002 | 0.998 )$	 &$2.71 ( 0.106 | 0.894 )$	&$2.946 ( 0.018 | 0.982 )$ &$3 ( 0 | 1 )$\\
&      &$\text{STP}_{2}$     &$2.944 ( 0.024 | 0.976 )$	 &$3 ( 0 | 1 )$		 &$3 ( 0 | 1 )$	 &$2.982 ( 0.006 | 0.994 )$	 &$3 ( 0 | 1 )$  &$3 ( 0 | 1 )$\\
&      &$\text{STP}_{3}$     &$3 ( 0 | 1 )$	 &$3 ( 0 | 1 )$		 &$3 ( 0 | 1 )$	 &$2.994 ( 0.002 | 0.998 )$	 &$3 ( 0 | 1 )$  &$3 ( 0 | 1 )$\\
&      &IterER               &$2.988 ( 0.004 | 0.996 )$	 &$3 ( 0 | 1 )$  &$3 ( 0 | 1 )$	 &$2.982 ( 0.006 | 0.994 )$ &$3 ( 0 | 1 )$  &$3 ( 0 | 1 )$\\
&      &$\alpha$-PCA         &$2.404 ( 0.2 | 0.8 )$	  &$2.494 ( 0.17 | 0.83 )$	 &$2.766 ( 0.084 | 0.902 )$	 &$2.862 ( 0.046 | 0.954 )$ &$2.94 ( 0.02 | 0.98 )$  &$2.934 ( 0.022 | 0.978 )$\\
&      &iTIP-IC               &$1 ( 1 | 0 )$	 &$1 ( 1 | 0 )$  &$1 ( 1 | 0 )$	 &$1 ( 1 | 0 )$ &$1 ( 1 | 0 )$  &$1 ( 1 | 0 )$\\
&      &iTIP-ER               &$2.646 ( 0.446 | 0.346 )$	 &$2.692 ( 0.448 | 0.312 )$  &$2.642 ( 0.436 | 0.376 )$	 &$2.736 ( 0.45 | 0.328 )$ &$2.924 ( 0.376 | 0.364 )$  &$2.782 ( 0.37 | 0.424 )$\\
&      &TCorTh         &$3.986 ( 0 | 0.014 )$	  &$3.988 ( 0 | 0.012 )$	 &$3.988 ( 0 | 0.012 )$	 &$3.988 ( 0 | 0.012 )$ &$4 ( 0 | 0 )$  &$3.998 ( 0 | 0.002 )$\\
&      &IC         &$2.994 ( 0.004 | 0.996 )$	  &$2.994 ( 0.006 | 0.994 )$	 &$3 ( 0 | 1 )$	 &$2.996 ( 0.002 | 0.998 )$ &$3 ( 0 | 1 )$  &$3 ( 0 | 1 )$\\
\cmidrule(lr){2-9}
&(3,3)   &$\text{STP}_{1}$     &$3 ( 0 | 1 )$	 &$3 ( 0 | 1 )$		 &$3 ( 0 | 1 )$	 &$3 ( 0 | 1 )$	 &$3 ( 0 | 1 )$		 &$3 ( 0 | 1 )$\\
&      &$\text{STP}_{2}$     &$3 ( 0 | 1 )$	 &$3 ( 0 | 1 )$		 &$3 ( 0 | 1 )$	 &$3 ( 0 | 1 )$	 &$3 ( 0 | 1 )$		 &$3 ( 0 | 1 )$\\
&      &$\text{STP}_{3}$     &$3 ( 0 | 1 )$	 &$3 ( 0 | 1 )$		 &$3 ( 0 | 1 )$	 &$3 ( 0 | 1 )$	 &$3 ( 0 | 1 )$  &$3 ( 0 | 1 )$\\
&      &IterER               &$3 ( 0 | 1 )$	 &$3 ( 0 | 1 )$  &$3 ( 0 | 1 )$	 &$3 ( 0 | 1 )$	 &$3 ( 0 | 1 )$  &$3 ( 0 | 1 )$\\
&      &$\alpha$-PCA         &$3 ( 0 | 1 )$	  &$3 ( 0 | 1 )$	 &$3 ( 0 | 1 )$	 &$3 ( 0 | 1 )$	  &$3 ( 0 | 1 )$	 &$3 ( 0 | 1 )$\\
&      &iTIP-IC               &$1 ( 1 | 0 )$	 &$1 ( 1 | 0 )$  &$1 ( 1 | 0 )$	 &$1 ( 1 | 0 )$	 &$1 ( 1 | 0 )$  &$1 ( 1 | 0 )$\\
&      &iTIP-ER               &$2.534 ( 0.426 | 0.478 )$	 &$2.69 ( 0.366 | 0.48 )$  &$2.71 ( 0.324 | 0.57 )$	 &$2.718 ( 0.304 | 0.606 )$	 &$2.824 ( 0.23 | 0.686 )$  &$2.842 ( 0.218 | 0.694 )$\\
&      &TCorTh         &$3.004 ( 0 | 0.996 )$	  &$3.002 ( 0 | 0.998 )$	 &$3.002 ( 0 | 0.998 )$	 &$3.012 ( 0 | 0.988 )$	  &$3 ( 0 | 1 )$	 &$3 ( 0 | 1 )$\\
&      &IC         &$8.458 ( 0.33 | 0.67 )$	  &$8.69 ( 0.234 | 0.766 )$	 &$8.908 ( 0.076 | 0.924 )$	 &$8.89 ( 0.064 | 0.936 )$	  &$8.972 ( 0.022 | 0.978 )$	 &$8.996 ( 0.004 | 0.996 )$\\
\bottomrule[2pt]
  \end{tabular*}}
  \end{threeparttable}
\end{table}
\vspace{1em}

\begin{table}[!h]
\caption{Simulation results for estimating $k_1$ in the form $x(y|z)$, $x$
is the sample mean of the estimated factor numbers based on 500 replications
$\protect\alpha =0.01,M=S=300,\protect\phi=\protect\psi=0.1,f(S)=S^{-1/4}$,
are the proportions of under- and exact- determination of the factor number,
respectively. }
\label{tab:main21}\renewcommand{\arraystretch}{1.5} \centering
\selectfont
\begin{threeparttable}
		\scalebox{0.65}{\begin{tabular*}{26cm}{ccccccccccccccccccccccccccccc}
				\toprule[2pt]
				&\multirow{2}{*}{$(k_1,k_2)$}&\multirow{2}{*}{Method}  &\multicolumn{3}{c}{$T=50$}&\multicolumn{3}{c}{$T=100$} \cr
				\cmidrule(lr){4-6} \cmidrule(lr){7-9}
				&&                 &$p_1=p_2=50$     &$p_1=p_2=100$       &$p_1=p_2=150$    &$p_1=p_2=50$     &$p_1=p_2=100$         &$p_1=p_2=150$ \\
				\midrule[1pt]
				&(1,1)   &$\text{STP}_{1}$     &$1.966 ( 0.014 | 0.006 )$	 &$1 ( 0 | 1 )$		 &$1 ( 0 | 1 )$	 &$1.976 ( 0.01 | 0.004 )$	&$2 ( 0 | 0 )$ &$1.05 ( 0 | 0.95 )$\\
				&      &$\text{STP}_{2}$     &$1.91 ( 0 | 0.09 )$	 &$1.002 ( 0 | 0.998 )$		 &$1 ( 0 | 1 )$	 &$1.694 ( 0 | 0.306 )$	 &$1.082 ( 0 | 0.918 )$  &$1 ( 0 | 1 )$\\
				&      &$\text{STP}_{3}$     &$1.028 ( 0 | 0.972 )$	 &$1 ( 0 | 1 )$		 &$1 ( 0 | 1 )$	 &$1 ( 0 | 1 )$	 &$1 ( 0 | 1 )$  &$1 ( 0 | 1 )$\\
				&      &IterER               &$1 ( 0 | 1 )$	 &$1 ( 0 | 1 )$  &$1 ( 0 | 1 )$	 &$1 ( 0 | 1 )$ &$1 ( 0 | 1 )$  &$1 ( 0 | 1 )$\\
				&      &$\alpha$-PCA         &$0.558 ( 0.53 | 0.382 )$	  &$0.88 ( 0.12 | 0.88 )$	 &$0.992 ( 0.008 | 0.992 )$	 &$1.144 ( 0.3 | 0.256 )$ &$0.952 ( 0.048 | 0.952 )$  &$1 ( 0 | 1 )$\\
				&      &iTIP-IC               &$1 ( 0 | 1 )$	 &$1 ( 0 | 1 )$  &$1 ( 0 | 1 )$	 &$1 ( 0 | 1 )$ &$1 ( 0 | 1 )$  &$1 ( 0 | 1 )$\\
				&      &iTIP-ER               &$1.258 ( 0 | 0.742 )$	 &$1.31 ( 0 | 0.69 )$  &$40.352 ( 0 | 0 )$	 &$1.266 ( 0 | 0.74 )$ &$1.236 ( 0 | 0.764 )$  &$1.236 ( 0 | 0.764 )$\\
				&      &TCorTh         &$2 ( 0 | 0 )$	  &$2 ( 0 | 0 )$	 &$1.55 ( 0 | 0.45 )$	 &$2 ( 0 | 0 )$ &$2 ( 0 | 0 )$  &$2 ( 0 | 0 )$\\
				&      &IC         &$1 ( 0 | 1 )$	  &$1 ( 0 | 1 )$	 &$1 ( 0 | 1 )$	 &$1 ( 0 | 1 )$ &$1 ( 0 | 1 )$  &$1 ( 0 | 1 )$\\
				\cmidrule(lr){2-9}
				&(1,3)   &$\text{STP}_{1}$     &$1.344 ( 0 | 0.656 )$	 &$1 ( 0 | 1 )$		 &$1 ( 0 | 1 )$	 &$1.294 ( 0 | 0.706 )$	 &$1.328 ( 0 | 0.672 )$ &$1 ( 0 | 1 )$\\
				&      &$\text{STP}_{2}$     &$1 ( 0 | 1 )$	 &$1 ( 0 | 1 )$		 &$1 ( 0 | 1 )$	 &$1 ( 0 | 1 )$	 &$1 ( 0 | 1 )$  &$1 ( 0 | 1 )$\\
		        &      &$\text{STP}_{3}$     &$1 ( 0 | 1 )$	 &$1 ( 0 | 1 )$		 &$1 ( 0 | 1 )$	 &$1 ( 0 | 1 )$	 &$1 ( 0 | 1 )$  &$1 ( 0 | 1 )$\\
				&      &IterER               &$1 ( 0 | 1 )$	 &$1 ( 0 | 1 )$  &$1 ( 0 | 1 )$	 &$1 ( 0 | 1 )$ &$1 ( 0 | 1 )$  &$1 ( 0 | 1 )$\\
				&      &$\alpha$-PCA         &$1 ( 0 | 1 )$	  &$1 ( 0 | 1 )$	 &$1 ( 0 | 1 )$	 &$1 ( 0 | 1 )$ &$1 ( 0 | 1 )$  &$1 ( 0 | 1 )$\\
				&      &iTIP-IC               &$1 ( 0 | 1 )$	 &$1 ( 0 | 1 )$  &$1 ( 0 | 1 )$	 &$1 ( 0 | 1 )$ &$1 ( 0 | 1 )$  &$1 ( 0 | 1 )$\\
				&      &iTIP-ER               &$1.192 ( 0 | 0.808 )$	 &$1.234 ( 0 | 0.766 )$  &$5.6 ( 0 | 0.724 )$	 &$1.084 ( 0 | 0.916 )$ &$1.076 ( 0 | 0.924 )$  &$1.104 ( 0 | 0.896 )$\\
				&      &TCorTh         &$2 ( 0 | 0 )$	  &$1.436 ( 0 | 0.564 )$	 &$1 ( 0 | 1 )$	 &$2 ( 0 | 0 )$ &$2 ( 0 | 0 )$  &$2 ( 0 | 0 )$\\
				&      &IC         &$2.78 ( 0.208 | 0.792 )$	  &$2.908 ( 0.09 | 0.91 )$	 &$2.904 ( 0.094 | 0.906 )$	 &$3 ( 0 | 1 )$ &$3 ( 0 | 1 )$  &$3 ( 0 | 1 )$\\
				\cmidrule(lr){2-9}
				&(3,1)   &$\text{STP}_{1}$     &$3.13 ( 0.002 | 0.866 )$	 &$3 ( 0 | 1 )$		 &$3 ( 0 | 1 )$	 &$3.086 ( 0.002 | 0.91 )$	&$3.182 ( 0 | 0.818 )$ &$3 ( 0 | 1  )$\\
				&      &$\text{STP}_{2}$     &$3.002 ( 0 | 0.998 )$	 &$3 ( 0 | 1 )$		 &$3 ( 0 | 1 )$	 &$3 ( 0 | 1 )$	 &$3 ( 0 | 1 )$  &$3 ( 0 | 1 )$\\
				&      &$\text{STP}_{3}$     &$3 ( 0 | 1 )$	 &$3 ( 0 | 1 )$		 &$3 ( 0 | 1 )$	 &$3 ( 0 | 1 )$	 &$3 ( 0 | 1 )$  &$3 ( 0 | 1 )$\\
				&      &IterER               &$3 ( 0 | 1 )$	 &$3 ( 0 | 1 )$  &$3 ( 0 | 1 )$	 &$3 ( 0 | 1 )$ &$3 ( 0 | 1 )$  &$3 ( 0 | 1 )$\\
				&      &$\alpha$-PCA         &$3.38 ( 0.142 | 0.072 )$	  &$3.054 ( 0.028 | 0.858 )$	 &$3 ( 0 | 1 )$	 &$3.768 ( 0.044 | 0.056 )$ &$3.084 ( 0.004 | 0.9 )$  &$3 ( 0 | 1 )$\\
				&      &iTIP-IC               &$1 ( 1 | 0 )$	 &$1 ( 1 | 0 )$  &$1 ( 1 | 0 )$	 &$1 ( 1 | 0 )$ &$1 ( 1 | 0 )$  &$1 ( 1 | 0 )$\\
				&      &iTIP-ER               &$2.556 ( 0.434 | 0.42 )$	 &$2.778 ( 0.34 | 0.492 )$  &$44.404 ( 0.02 | 0.058 )$	 &$2.572 ( 0.45 | 0.386 )$ &$2.818 ( 0.316 | 0.496 )$  &$2.87 ( 0.268 | 0.57 )$\\
				&      &TCorTh         &$3.844 ( 0 | 0.156 )$	  &$3.052 ( 0 | 0.948 )$	 &$3 ( 0 | 1 )$	 &$3.992 ( 0 | 0.008 )$ &$3.998 ( 0 | 0.002 )$  &$3.784 ( 0 | 0.216 )$\\
				&      &IC         &$2.776 ( 0.21 | 0.79 )$	  &$2.866 ( 0.132 | 0.868 )$	 &$2.92 ( 0.078 | 0.922 )$	 &$3 ( 0 | 1 )$ &$3 ( 0 | 1 )$  &$3 ( 0 | 1 )$\\
				\cmidrule(lr){2-9}
				&(3,3)   &$\text{STP}_{1}$     &$3 ( 0 | 1 )$	 &$3 ( 0 | 1 )$		 &$3 ( 0 | 1 )$	 &$3 ( 0 | 1 )$	 &$3 ( 0 | 1 )$		 &$3 ( 0 | 1 )$\\
				&      &$\text{STP}_{2}$     &$3 ( 0 | 1 )$	 &$3 ( 0 | 1 )$		 &$3 ( 0 | 1 )$	 &$3 ( 0 | 1 )$	 &$3 ( 0 | 1 )$		 &$3 ( 0 | 1 )$\\
			    &      &$\text{STP}_{3}$     &$3 ( 0 | 1 )$	 &$3 ( 0 | 1 )$		 &$3 ( 0 | 1 )$	 &$3 ( 0 | 1 )$	 &$3 ( 0 | 1 )$		 &$3 ( 0 | 1 )$\\
				&      &IterER               &$3 ( 0 | 1 )$	 &$3 ( 0 | 1 )$  &$3 ( 0 | 1 )$	 &$3 ( 0 | 1 )$	 &$3 ( 0 | 1 )$  &$3 ( 0 | 1 )$\\
				&      &$\alpha$-PCA         &$2.998 ( 0.002 | 0.998 )$	  &$3 ( 0 | 1 )$	 &$3 ( 0 | 1 )$	 &$2.998 ( 0.002 | 0.998 )$	  &$3 ( 0 | 1 )$	 &$3 ( 0 | 1 )$\\
				&      &iTIP-IC               &$1 ( 1 | 0 )$	 &$1 ( 1 | 0 )$  &$1 ( 1 | 0 )$	 &$1 ( 1 | 0 )$	 &$1 ( 1 | 0 )$  &$1 ( 1 | 0 )$\\
				&      &iTIP-ER               &$2.61 ( 0.408 | 0.456 )$	 &$2.806 ( 0.282 | 0.578 )$  &$5.514 ( 0.254 | 0.54 )$	 &$2.728 ( 0.284 | 0.632 )$	 &$2.824 ( 0.228 | 0.676 )$  &$2.906 ( 0.166 | 0.744 )$\\
				&      &TCorTh         &$3.002 ( 0 | 0.998 )$	  &$3 ( 0 | 1 )$	 &$3 ( 0 | 1 )$	 &$3.102 ( 0 | 0.898 )$	  &$3 ( 0 | 1 )$	 &$3 ( 0 | 1 )$\\
				&      &IC         &$7.192 ( 0.888 | 0.112 )$	  &$7.832 ( 0.712 | 0.288 )$	 &$7.896 ( 0.74 | 0.26 )$	 &$8.958 ( 0.038 | 0.962 )$	  &$8.998 ( 0.002 | 0.998 )$	 &$9 ( 0 | 1 )$\\
				\bottomrule[2pt]
		\end{tabular*}}
	\end{threeparttable}
\end{table}

\vspace{1em}

Results are in Tables \ref{tab:main2} and \ref{tab:main21}, from which we
can draw the following three conclusions.

First, especially for the case $k_{1}=k_{2}=1$, and especially when $p_{2}$
is small, the $\text{STP}_{2}$ and STP$_{3}$ procedures dominate the IterER
and the $\alpha $-PCA procedures, which have a pronounced tendency to
understate the number of common factors - thus, in this case, mistakenly
finding no evidence of a row factor structure and, consequently, mistakenly
indicating a vector, as opposed to a matrix, factor model. This is not true
for the iTIP-IC and the iTIP-ER procedures, which always correctly estimate $%
k_{1}$ as equal to $1$; however, these procedures (especially iTIP-IC) are
less able to determine the presence of further common factors when $k_{1}=3$%
. Note further that, by construction, they are initialised at $k_{1}=1$, so
they cannot understate $k_{1}$ but are unable to check whether $k_{1}=0$
(although, in principle, it is possible to extend this method to check if
there are $k_{1}=0$ factors by constructing an artificial eigenvalue as in %
\citealp{ahnhorenstein13}). In general, the $\text{STP}_{2}$ procedure (and,
to a lesser extent, the STP$_{3}$ one) dominates over all other procedures
in almost all cases considered, which makes a very strong case to consider
the preliminary step of projecting the data $X_{t}$ onto $C$ prior to
undertaking any analysis. We also note that the $\text{STP}_{1}$ method
performs comparably with the $\alpha $-PCA method, but it is inferior to the
IterER method, albeit with some exceptions - e.g. when $k_{1}=k_{2}=1$, and $%
p_{2}=15$. Results, as Table \ref{tab:main21} demonstrates, improve as the
sample sizes increase, and become comparable with those obtained with other
criteria. The results in Table \ref{tab:main21} also contain the case of
smaller $p_{1}=50$: in such a case, it is evident that reducing $p_{1}$
worsens the overall ability of our procedures, which can be explained by
noting that a lower $p_{1}$ corresponds to the spiked eigenvalues diverging
at a slower rate.

Second, all procedures seem to improve as $k_{2}$ increases, as can be
anticipated in the light of (\ref{stn}). In such cases, the $\text{STP}_{2}$
and STP$_{3}$\ procedures still retain their advantage especially for small
values of $p_{2}$, but less evidently than in the previous cases.

Third, confirming what is also found in the previous section, the $\text{STP}%
_{2}$ and STP$_{3}$ methodologies always outperform the $\text{STP}_{1}$
one. Indeed, the $\text{STP}_{1}$ procedure works well for large sample
sizes, but it is dominated even by the IterER methodology in small samples
(with few exceptions), and, to a lesser extent, by the $\alpha $-PCA one.
This suggests that the gains observed with $\text{STP}_{2}$ arise from two
equally important sources: the use of the projection method proposed by %
\citet{hkyz2021}, and the use of our randomised tests \textit{cum} the
decision rule advocated in (\ref{thresholds}).

Finally, we note that we have run further experiments in the Supplementary
Material to assess the robustness of our procedures. In Table \ref{tab:r3q72}%
, we consider the case of $p_{1}=p_{2}$ and \textquotedblleft
small\textquotedblright , obtaining, broadly, the same results as above; the
case of weak factors is studied in Tables \ref{tab:weakf1} and \ref%
{tab:weakf2}; finally, in Table \ref{tab:k2zero}, we analyse the case in
which $k_{2}=0$, but the applied user mistakenly uses $\widetilde{k}_{1}$.
We also investigate the robustness of our procedure to its specifications.
In Table \ref{weight-choice}, we assess the robustness to different
specifications to the weight function $F\left( u\right) $ required in Step 4
of our randomisation algorithm, and in Table \ref{tab:supp4}, we assess the
impact of the threshold function $f\left( S\right) $ for the case $%
(k_{1},k_{2})=(3,3)$. The broad conclusion, even in this case, is that our
procedure is not affected by these specifications, reinforcing the message
that although some specifications need to be chosen by the researcher, the
impact thereof is negligible. Finally, in Tables \ref{tab:CT1}-\ref{tab:CT3}
we report the computational times of our procedures, comparing them against
those of the other criteria considered in the above.

\section{Empirical studies\label{empirical}}

We illustrate our procedure through two applications: we firstly present an
application to a set of macroeconomic indices (Section \ref{mindex}), and
then consider a 2-D image recognition dataset (Section \ref{mnist}).

\subsection{Multinational macroeconomic indices\label{mindex}}

Inspired by \citet{fan2021}, we investigate the presence (and dimension) of
a matrix factor structure in a time series of macroeconomic indices. In our
application, we use the dataset employed by \cite{hkyz2021}, containing
records of $p_{2}=10$ macroeconomic indices across $p_{1}=8$ OECD\ countries
over $T=130$ quarters, ranging from $1988Q1$ to $2020Q2$. Whilst we refer to
\cite{hkyz2021} for details, the countries are the United States, the United
Kingdom, Canada, France, Germany, Norway, Australia and New Zealand, which
can be naturally divided into three groups as North American, European and
Oceania based on their geographical locations. The indices are from $4$
major groups, namely consumer price, interest rate, production, and
international trade.\footnote{%
In particular, we have considered the following indices, grouped by family:
productivity (comprising: Total Index excluding Construction, Total
Manifacturing, and GDP), CPI (comprising Food, Energy, and CPI Total),
interest rates (long-term government bond yields, and 3-month Interbank
rates and yields), and international trade (comprising total exports and
total imports, both measured by value).} As in \cite{hkyz2021}, we use the
log-differences of each index, and each series is standardised.

\bigskip

We begin with testing whether there exists a matrix factor structure in the
data. Results are in Tables \ref{tab:test1} and \ref{tab:test2}; as also
found in \citet{fan2021}, there is overwhelming evidence in favour of a
matrix structure in the data, for all test specifications considered, which
corresponds to not rejecting the null hypotheses that $k_{1},k_{2}\geq 1$.

\begin{table}[htbp]
\caption{Testing the null hypothesis $H_{01}:k_1\ge 1$ for the macroeconomic
indices data set. Tests are based on $\protect\alpha=0.01, M=100$, using the
following thresholds: $f_1(S)=1-\protect\alpha-\protect\sqrt{2\ln S/S}$, $%
f_2(S)=1-\protect\alpha-S^{-1/3}$, $f_3(S)=1-\protect\alpha-S^{-1/4}$, $%
f_4(S)=1-\protect\alpha-S^{-1/5}$, $f_5(S)=(1-\protect\alpha)/2$.}
\label{tab:test1}\renewcommand{\arraystretch}{1.5} \centering
\selectfont
\begin{threeparttable}
		 \scalebox{0.9}{\begin{tabular*}{18cm}{ccccccccccccccccccccccccccccc}
				\toprule[2pt]
				 &&\multicolumn{5}{c}{No-projection}&\multicolumn{5}{c}{Projection}\\\cmidrule(lr){2-6}\cmidrule(lr){7-11}
				$S$&$f_1(S)$&	$f_2(S)$&	 $f_3(S)$&	$f_4(S)$&	$f_5(S)$&$f_1(S)$&	 $f_2(S)$&	$f_3(S)$&	$f_4(S)$&	$f_5(S)$\\	 \midrule
200&Accept&Accept&Accept&Accept&Accept&Accept&Accept&Accept&Accept&Accept\\
300&Accept&Accept&Accept&Accept&Accept&Accept&Accept&Accept&Accept&Accept\\
400&Accept&Accept&Accept&Accept&Accept&Accept&Accept&Accept&Accept&Accept\\
				\bottomrule[2pt]
		\end{tabular*}}
	\end{threeparttable}
\end{table}
\begin{table}[!h]
\caption{Testing the null hypothesis $H_{01}:k_2\ge 1$ for the macroeconomic
indices data set. Tests are based on $\protect\alpha=0.01, M=100$, using the
following thresholds: $f_1(S)=1-\protect\alpha-\protect\sqrt{2\ln S/S}$, $%
f_2(S)=1-\protect\alpha-S^{-1/3}$, $f_3(S)=1-\protect\alpha-S^{-1/4}$, $%
f_4(S)=1-\protect\alpha-S^{-1/5}$, $f_5(S)=(1-\protect\alpha)/2$.}
\label{tab:test2}\renewcommand{\arraystretch}{1.5} \centering
\selectfont
\begin{threeparttable}
		 \scalebox{0.9}{\begin{tabular*}{18cm}{ccccccccccccccccccccccccccccc}
				\toprule[2pt]
				 &&\multicolumn{5}{c}{No-projection}&\multicolumn{5}{c}{Projection}\\\cmidrule(lr){2-6}\cmidrule(lr){7-11}
				$S$&$f_1(S)$&	$f_2(S)$&	 $f_3(S)$&	$f_4(S)$&	$f_5(S)$&$f_1(S)$&	 $f_2(S)$&	$f_3(S)$&	$f_4(S)$&	$f_5(S)$\\	 \midrule
200&Accept&Accept&Accept&Accept&Accept&Accept&Accept&Accept&Accept&Accept\\
300&Accept&Accept&Accept&Accept&Accept&Accept&Accept&Accept&Accept&Accept\\
400&Accept&Accept&Accept&Accept&Accept&Accept&Accept&Accept&Accept&Accept\\
				\bottomrule[2pt]
		\end{tabular*}}
	\end{threeparttable}
\end{table}

We now turn to determining the dimensions of the row and column factor
spaces $k_{1}$ and $k_{2}$. The empirical exercise in \citet{fan2021}
demonstrates that, possibly owing to the small cross-sectional sizes, the
estimated number of common factors differs considerably depending on the
estimation method employed. Table \ref{tab:mac1} reports the estimated
numbers of common factors in the form of $(a,b)$ where $a$ denotes the
number of common row factors, and $b$ the number of column factors. Using
the results in Section \ref{simulation} as guidelines (see in particular
Section \ref{diffp1p2}), we have used $M=150$, $S=200$ and $\alpha =0.01$;
by way of robustness check, we have also considered different values of $%
\alpha $, noting that, as $\alpha $ increases, our proposed procedure is
more and more in favor of rejecting the existence of factors, and we have
changed $M$ to $M=250$, obtaining the same results as reported here.%
\footnote{%
Unreported results show that, for $\alpha \geq 0.05$, one would find $%
k_{1}=1 $ and $k_{2}=0$, thus rejecting a matrix factor structure
altogether. This reinforces the findings in the previous section, where it
was noted that our procedure, in small samples, requires a smaller $\alpha $
in order to estimate the factor dimensions correctly.} Here, we report
results using different combinations of $S$ and $f\left( S\right) $ to shed
further light on the impact of these specifications; in particular, we use:
the thresholds employed also in Section \ref{simulation} - i.e. $f\left(
S\right) =S^{-a}$ with $a={1}/{3}$, ${1}/{4}$ and ${1}/{5}$; a conservative
threshold, $f\left( S\right) =\sqrt{2\ln S/S}$; and a very \textquotedblleft
liberal\textquotedblright\ one, with $f\left( S\right) =\left( 1-\alpha
\right) /2$.

\bigskip

\begin{table}[htbp]
\caption{Estimated numbers of row and column factors for the macroeconomic
indices data set. Five ways to select the threshold: $f_1(S)=1-\protect\alpha%
-\protect\sqrt{2\ln S/S}$, $f_2(S)=1-\protect\alpha-S^{-1/3}$, $f_3(S)=1-%
\protect\alpha-S^{-1/4}$, $f_4(S)=1-\protect\alpha-S^{-1/5}$, $f_5(S)=(1-%
\protect\alpha)/2$. }
\label{tab:mac1}\renewcommand{\arraystretch}{1.5} \centering
\selectfont
\begin{threeparttable}
		\scalebox{1}{\begin{tabular*}{7.8cm}{cccccc}
\toprule[2pt] & \multicolumn{5}{c}{No-projection} \\
 $S$ & $f_{1}(S)$ & $f_{2}(S)$ & $%
f_{3}(S)$ & $f_{4}(S)$ & $f_{5}(S)$ \\
\midrule200 & (1,4) & (1,4) & (1,4) & (1,4) & (2,4) \\
300 & (1,3) & (1,0) & (1,4) & (1,4) & (2,4) \\
400 & (1,3) & (1,3) & (1,4) & (1,4) & (2,4) \\ \hline
& \multicolumn{5}{c}{Projection} \\
$S$ & $f_{1}(S)$ & $f_{2}(S)$ & $f_{3}(S)$ & $f_{4}(S)$ & $f_{5}(S)$ \\
\midrule200  & (1,4) & (1,3) & (1,4) & (2,4) & (2,4) \\
300 & (1,3) & (1,3) & (1,3) & (1,4) & (2,4) \\
400 & (1,3) & (1,3) & (1,3) & (1,4) & (2,4) \\
\bottomrule[2pt] &  &  &  &  &
\end{tabular*}}
	\end{threeparttable}
\end{table}

\bigskip

Results are only partly affected by the choice of $S$ and $f\left( S\right) $%
, which play a very minor role (a desirable form of robustness). As pointed
out in Section \ref{simulation}, the projection technique should work better
in finite samples, but in our application results are actually comparable
between the two techniques. According to Table \ref{tab:mac1}, the number of
row factors is at most $k_{1}=2$: whilst there is strong evidence in favour
of at least one common factor (thus confirming that there is a matrix factor
structure, as also found by \citet{fan2021} using the eigenvalue ratio
approach), the second factor seems to be weaker, and deciding whether $%
\widehat{k}_{1}=1$ or $2$ can be done on account of the researcher's
preference for (possible) underestimation versus overestimation. Reading
these results in conjunction with Table 9 in \cite{hkyz2021} would suggest
choosing $\widehat{k}_{1}=2$: factors broadly represent the different
geographical locations, but European countries (particularly the largest
economy, Germany) seem to also share a common factor structure with North
America, speaking to the integration between the two economic areas. As far
as $\widehat{k}_{2}$ is concerned, using the majority vote when applying the
projection technique suggests $\widehat{k}_{2}=4$; even in this case, there
seems to be some evidence in favour of $\widehat{k}_{2}=3$ also, again
suggesting that, possibly, the fourth common factor is weaker than the
others. Interestingly, the results in \citet{fan2021} using two different
techniques (respectively, a scree-plot and an eigenvalue ratio approach)
indicate that $k_{2}$ may range between $2$ and $6$, so our proposed
approach offers a considerable refinement; \cite{hkyz2021} also find $%
\widehat{k}_{2}=4$ or $5$, but their results with $\widehat{k}_{2}=4$ show
that this estimate explains the data very well, and it matches the four
groups to which the indices belong which is an intuitive and meaningful
finding.

Finally, by way of comparison we report the estimated values of the number
of row and column common factors using various techniques available in the
literature; results are in Table \ref{tab:commacro}. There is broad
consensus across all techniques as far as $\widehat{k}_{1}$ is concerned -
one common factor is found by virtually all criteria, with the exception of
the iTIP-ER criterion which indicates $\widehat{k}_{1}=2$. This is in line
with our estimates, which (as mentioned above) suggest the presence of one
strong common factor, and also the possible presence of a (weaker) second
common factor. Conversely, there seems to be less consensus when estimating $%
k_{2}$. All criteria indicate a small value of $\widehat{k}_{2}$, which,
also in the light of the empirical exercise in \cite{hkyz2021}, seems to be
an understatement of the true number of common factors. As mentioned above,
the $\alpha $-PCA\ criterion delivers very different values of $\widehat{k}%
_{2}$ depending on the value of $\alpha $ (in our case, we have used $\alpha
=0$ as in the Monte Carlo exercise); the iterative procedures by %
\citet{han2020rank} indicate that $\widehat{k}_{2}=1$ or $2$, thus
confirming the findings in Table \ref{tab:r3q72} which suggest a tendency to
understate the number of common factors in small samples when this is larger
than one. The criterion by \citet{lam2021rank}, on the other hand, is the
one closest to our findings, indicating that $\widehat{k}_{2}=3$.

\begin{table}[htbp]
\caption{Estimated numbers of row and column factors (or total number for
vectorized methods) using different approaches in the literature, for the
macroeconomic indices data set. $k_{\max}$ is set as $6$ for
matrix-factor-model based approaches and $6^2$ for vectorized methods. }
\label{tab:commacro}\renewcommand{\arraystretch}{1.5} \centering
\selectfont
\begin{threeparttable}
		\scalebox{1}{\begin{tabular*}{16.5cm}{ccccccccccc}
				\toprule[2pt]
				&&&\multicolumn{4}{c}{tensorTS} &&\multicolumn{2}{c}{Vectorized} \\\cmidrule(lr){4-7}\cmidrule(lr){9-10}
				& IterER &$\alpha$-PCA&iTIP-IC&iTOP-IC&iTIP-ER&iTOP-ER &TCorTh & ER& IC\\\midrule[1.2pt]
				$\hat k_1$&1&1&1&1&2&1&1&\multirow{2}{*}{1}&\multirow{2}{*}{2}
\\
				$\hat k_2$&5&2&1&1&2&2&3&&\\
				\bottomrule[2pt]
		\end{tabular*}}
	\end{threeparttable}
\end{table}

\begin{comment}

We have tried different combinations of $S$, $\alpha $ and threshold $f(S)$ while $M=100$. It's seen that the results are most sensitive to $%
\alpha $, while the effects of $S$ and the thresholding rules are less
important. When both $\alpha $ is large, our procedure is more in favour of
rejecting the existence of factors so that most of the results are $(1,0)$.
This is partially due to the small $(p_{1},p_{2})$ setting for this real
example, which supplements to our simulation studies. Hinted by the
simulation results, we prefer to accept the results with $\alpha =0.01$,
then at most we find 2 factors for the countries and 4 factors for the
macroeconomic indices. This result is consistent with that in \citet{hkyz2021}, and roughly matches with the group numbers of both
dimensions. Overall, we suggest taking $k_{1}=2$ (or $k_{1}=3$) and $k_{2}=4$
in case of underestimation since the second and third factors seem to be
\textquotedblleft weak\textquotedblright . Recall that non-zero
estimates of the factor numbers are equivalent to accepting the existence of
common factors. Therefore, this study rationales the matrix-factor-structure
assumption in \citet{hkyz2021}.
\end{comment}%
\begin{comment}
\begin{table}[htbp]
\caption{Estimated numbers of row and column factors for macroeconomic
indices data set. Five ways to select the threshold: $f_1(S)=1-\protect\alpha%
-\protect\sqrt{2\ln S/S}$, $f_2(S)=1-\protect\alpha-S^{-1/3}$, $f_3(S)=1-%
\protect\alpha-S^{-1/4}$, $f_4(S)=1-\protect\alpha-S^{-1/5}$, $f_5(S)=(1-%
\protect\alpha)/2$. }
\label{tab:mac}\renewcommand{\arraystretch}{1.5} \centering
\selectfont
\begin{threeparttable}
		 \scalebox{1}{\begin{tabular*}{16cm}{ccccccccccccccccccccccccccccc}
				\toprule[2pt]
				 &&\multicolumn{5}{c}{No-projection}&\multicolumn{5}{c}{Projection}\\\cmidrule(lr){3-7}\cmidrule(lr){8-12}
		$S$&$\alpha$&$f_1(S)$&	$f_2(S)$&	 $f_3(S)$&	$f_4(S)$&	$f_5(S)$&$f_1(S)$&	 $f_2(S)$&	$f_3(S)$&	$f_4(S)$&	$f_5(S)$\\	 \midrule
200&0.01&(1,4)&(1,4)&(1,4)&(1,4)&(2,4)&(1,4)&(1,3)&(1,4)&(2,4)&(2,4)
\\
300&0.01&(1,3)&(1,0)&(1,4)&(1,4)&(2,4)&(1,3)&(1,3)&(1,3)&(1,4)&(2,4)
\\
400&0.01&(1,3)&(1,3)&(1,4)&(1,4)&(2,4)&(1,3)&(1,3)&(1,4)&(1,4)&(2,4)
\\\hline
200&0.05&(1,0)&(1,0)&(1,0)&(1,0)&(1,0)&(1,0)&(1,0)&(1,0)&(1,0)&(1,0)
\\
300&0.05&(1,0)&(1,0)&(1,0)&(1,0)&(1,0)&(1,0)&(1,0)&(1,0)&(1,0)&(1,0)
\\
400&0.05&(1,0)&(1,0)&(1,0)&(1,0)&(1,0)&(1,0)&(1,0)&(1,0)&(1,0)&(1,0)
\\\hline
200&0.1&(1,0)&(0,0)&(1,0)&(1,0)&(1,0)&(1,0)&(1,0)&(1,0)&(1,0)&(1,0)
\\
300&0.1&(1,0)&(0,0)&(1,0)&(1,0)&(1,0)&(1,0)&(0,0)&(1,0)&(1,0)&(1,0)
\\
400&0.1&(0,0)&(0,0)&(1,0)&(1,0)&(1,0)&(0,0)&(0,0)&(1,0)&(1,0)&(1,0)\\
				\bottomrule[2pt]
		\end{tabular*}}
	\end{threeparttable}
\end{table}
\end{comment}

\subsection{MNIST: handwritten digit numbers\label{mnist}}

In our second example, we apply matrix time series to an image recognition
dataset, namely the Modified National Institute of Standards and Technology
(MNIST) dataset, which has been analysed in numerous applications of
classification algorithms and machine learning, and which consists of images
of handwritten digit numbers from $0$ to $9$. As is typical in these
applications, each single (gray-scale) image represents the matrix $X_{t}$,
whose elements are the pixels of the image. We only use the training set,
which contains $T=10,000$ images; in our dataset, the digits have been
size-normalized and centered in a fixed-size image with $28\times 28$
pixels, thus having $p_{1}=p_{2}=28$. We standardize the pixels at each
location.

\bigskip

The estimated numbers of row and column factors are reported in Table \ref%
{tab:mnist} with multiple combinations of $\alpha $, $S$ and $f\left(
S\right) $, as in the previous section. Since $p_{1}$ and $p_{2}$ are larger
than those in our previous example, we use $M=200$ in the testing (we tried $%
M=100$ and $300$ and results are essentially the same).

\begin{table}[htbp]
\caption{Estimated numbers of row and column factors for handwritten digit
number data set. Five ways to select the threshold. Q1: $1-\protect\alpha-%
\protect\sqrt{2\ln\ln S/S}$; Q2: $1-\protect\alpha-S^{-1/3}$; Q3: $1-\protect%
\alpha-S^{-1/4}$; Q4: $1-\protect\alpha-S^{-1/5}$; Q5: $(1-\protect\alpha)/2$%
. }
\label{tab:mnist}\renewcommand{\arraystretch}{1.5} \centering
\selectfont
\begin{threeparttable}
		 \scalebox{1}{\begin{tabular*}{16.5cm}{ccccccccccccccccccccccccccccc}
				\toprule[2pt]
				 &&\multicolumn{5}{c}{No-projection}&\multicolumn{5}{c}{Projection}\\\cmidrule(lr){3-7}\cmidrule(lr){8-12}
		$S$&$\alpha$&$f_1(S)$&	$f_2(S)$&	 $f_3(S)$&	$f_4(S)$&	$f_5(S)$&$f_1(S)$&	 $f_2(S)$&	$f_3(S)$&	$f_4(S)$&	$f_5(S)$\\	 \midrule
200&0.01&(0,3)&(0,3)&(0,3)&(0,4)&(4,5)&(4,5)&(4,5)&(4,5)&(4,5)&(4,5)
\\
300&0.01&(0,3)&(0,3)&(0,3)&(0,4)&(4,5)&(4,5)&(4,4)&(4,5)&(4,5)&(4,5)
\\
400&0.01&(0,3)&(0,3)&(0,3)&(0,3)&(4,5)&(4,5)&(4,5)&(4,5)&(4,5)&(4,5)
\\\hline
200&0.05&(0,0)&(0,0)&(0,0)&(0,1)&(0,1)&(4,3)&(4,3)&(4,3)&(4,3)&(4,3)
\\
300&0.05&(0,0)&(0,0)&(0,0)&(0,1)&(0,1)&(4,3)&(4,3)&(4,3)&(4,3)&(4,3)
\\
400&0.05&(0,0)&(0,0)&(0,1)&(0,1)&(0,1)&(4,3)&(4,3)&(4,3)&(4,3)&(4,3)
\\\hline
200&0.1&(0,0)&(0,0)&(0,0)&(0,0)&(0,0)&(4,3)&(0,3)&(4,3)&(4,3)&(4,3)
\\
300&0.1&(0,0)&(0,0)&(0,0)&(0,0)&(0,0)&(0,3)&(0,3)&(4,3)&(4,3)&(4,3)
\\
400&0.1&(0,0)&(0,0)&(0,0)&(0,0)&(0,0)&(0,3)&(0,3)&(0,3)&(4,3)&(4,3)\\
				\bottomrule[2pt]
		\end{tabular*}}
	\end{threeparttable}
\end{table}

Results and conclusions are similar to those in the previous section. In
particular, a bigger difference emerges in the performance of projection
versus non-projection based estimation, with the former offering a
performance which is more robust across the different specifications. In
light also of the results in Section \ref{simulation}, the findings in this
section strengthen the case in favour of the projection-based estimator. We
note that, when using this technique, the number of row factors is almost
always (save for some exceptions, based on a large $\alpha $ and a high
threshold $f\left( S\right) $) estimated as $\widehat{k}_{1}=4$. As far as $%
k_{2}$ is concerned, all results indicate that this is not smaller than $3$,
and the most conservative approach (based on using $\alpha =0.01$) indicates
the possibility of having $\widehat k_{2}=5$. This may suggest that two
common factors are less pervasive than the others. In order to avoid
underestimation, we recommend taking $\widehat{k}_{1}=4$ and $\widehat{k}%
_{2}=5$ in this example. Indeed, in any real applications, we suggest the
readers to try different combinations of $\alpha $ and threshold, and select
the numbers of factors based on the real tolerance of underestimation and
overestimation errors. Smaller $\alpha $ and threshold are in favour of $%
H_{0} $, but in higher risk of overestimation. Larger $\alpha $ and
threshold will lead to opposite results.

For this example, we further compare the results for different digit numbers
in Table \ref{tab:compare} using only a small part of the images associated
with a specific number. In this table, we report results corresponding to $%
S=400$, $\alpha =0.01$ and $f_{5}(S)$; we point out however that using
different specifications leaves the results virtually unchanged. Results are
remarkably stable across the different digits.

\begin{table}[htbp]
\caption{Estimated numbers of row and column factors for different digit
numbers. }
\label{tab:compare}\renewcommand{\arraystretch}{1.5} \centering
\selectfont
\begin{threeparttable}
		 \scalebox{1}{\begin{tabular*}{15.5cm}{ccccccccccccccccccccccccccccc}
				\toprule[2pt]
				Projection& ``0''& ``1''& ``2''& ``3''& ``4''& ``5''& ``6''& ``7''& ``8''& ``9''\\\midrule
No&(4,5)&(4,5)&(4,5)&(4,5)&(4,5)&(4,5)&(4,5)&(4,5)&(4,5)&(4,4)
\\
Yes&(4,5)&(4,5)&(4,5)&(4,5)&(4,5)&(4,5)&(4,5)&(4,5)&(4,5)&(4,5)\\
				\bottomrule[2pt]
		\end{tabular*}}
	\end{threeparttable}
\end{table}

Finally, similarly to the previous application, we compare our results
against those obtained using alternative criteria. The results in Table \ref%
{tab:comdigit} show that results are broadly similar across the various
techniques, in a more evident way than in the case of the previous exercise.
In particular, all criteria indicate $\widehat{k}_{1}\geq 4$. The iterative
procedures by \citet{han2020rank} show the same pattern as before, with
Information Criteria having a tendency to estimate a larger number of common
factors than the Eigenvalue Ratio statistic. Interestingly, the criteria
proposed by \citet{lam2021rank} seem to overstate the number of common
factors - this is particularly evident when comparing $\widehat{k}_{1}$,
which the majority of criteria finds to be equal to $4$, and it is found to
be equal to $8$ using the estimator by \citet{lam2021rank}. As far as $%
\widehat{k}_{2}$ is concerned, the consensus is that $\widehat{k}_{2}\geq 3$%
, with the majority vote agreeing with our estimate that $\widehat{k}_{2}=5$.

\begin{table}[htbp]
\caption{Estimated numbers of row and column factors (or total number for
vectorized methods) using different approaches in the literature, for
different digit numbers. $k_{\max }$ is set as $10$ for matrix-factor-model
based approaches and $10^{2}$ for vectorized methods. }
\label{tab:comdigit}\renewcommand{\arraystretch}{1.5} \centering
\selectfont
\begin{threeparttable}
		\scalebox{1}{\begin{tabular*}{17cm}{ccccccccccc}
				\toprule[2pt]
				&&&\multicolumn{4}{c}{tensorTS} &&\multicolumn{2}{c}{Vectorized} \\\cmidrule(lr){4-7}\cmidrule(lr){9-10}
				& IterER &$\alpha$-PCA&iTIP-IC&iTOP-IC&iTIP-ER&iTOP-ER &TCorTh & ER& IC\\\midrule[1.2pt]
				$\hat k_1$&4&4&6&6&4&4&8&\multirow{2}{*}{4}&\multirow{2}{*}{12}\\
				$\hat k_2$&5&1&5&6&3&5&6&&\\
				\bottomrule[2pt]
		\end{tabular*}}
	\end{threeparttable}
\end{table}

\section{Discussion and conclusions\label{conclusion}}

In this contribution, we studied the important issue of determining the
presence and dimension of the row and column factor structures in a series
of matrix-valued data exhibiting a Kronecker product structure in the
loadings. Our methodology allows to check whether there is a factor
structure in either dimension (row and column), thus helping the researcher
decide whether data should be studied using the techniques developed by the
literature for a standard vector factor model, or whether different
techniques should be employed that are specific to tensor-valued data. In
addition to finding evidence of a factor structure, we also proposed a
methodology to estimate the numbers of common row and column factors.

Technically, our methodology is based on exploiting the eigen-gap which is
found, in the presence of common factors, in the sample second moment matrix
of the series. For each eigenvalue, we propose a test for the null that it
diverges (as opposed to being bounded). Our tests are similar to the
randomised tests (designed for vector factor models) proposed in %
\citet{trapani2018randomized}. However, we substantially refine rates via a
different method of proof, and (crucially) we propose a \textquotedblleft
strong\textquotedblright , Law-of-the-Iterated-Logarithm-inspired, decision
rule which does away with the randomness, thus ensuring that all researchers
using the same datasets will obtain the same results. In our paper, we
proposed two procedures, based on two different ways of computing the sample
second moment matrix: specifically, we use a \textquotedblleft
flattened\textquotedblright\ version of the matrix-valued series, and a
projected version thereof, as proposed in \citet{hkyz2021}. We found that
both techniques work very well in large samples, but, in small samples, the
projection-based method is superior in all scenarios considered, also
outperforming other existing methods.

\bigskip

Several important issues remain outstanding. In particular, from the outset,
we have assumed that model (\ref{kron}) + (\ref{fm1}) is correct, i.e. that
the loading space has a Kronecker product structure. As we discussed in the
introduction, under this assumption the separate estimation of the loadings
matrices $C$ and $R$ is advantageous since it entails a substantial
dimensionality reduction: under (\ref{kron}), the estimation of $%
p_{1}k_{1}+p_{2}k_{2}$ coefficients is required, compared to estimating $%
\Lambda $ in (\ref{bai03}), which contains $p_{1}p_{2}k_{1}k_{2}$
coefficients. Moreover, $C$ and $R$ have a clear interpretation, and
estimating them allows to understand the interplay between the row factors
and column factors, whereas $\Lambda $ does not allow for such an
interpretation. However, all these advantages are predicated on (\ref{kron})
being correctly specified in the first place. If this is the case, it would
be possible to construct some pathological counterexamples in which (\ref%
{bai03}) is correct, whereas (\ref{kron}) is not, and - when mistakenly
using (\ref{kron}) and the techniques proposed in this paper - our tests
find $k_{1}=0$ and $k_{2}=0$ even when the dimension of the factor space in (%
\ref{bai03}), $k$, is strictly positive.\footnote{%
We are grateful to an anonymous Referee for pointing this out to us.} Whilst
this issue goes beyond the scope of the present paper, we offer a more
in-depth discussion of this issue through an example. Consider the case of $%
p_{2}<p_{1}$ and consider the following vector factor model%
\begin{equation}
\text{Vec}\left( X_{t}\right) =\Lambda F_{t}+\text{Vec}\left( E_{t}\right) .
\label{vec-kron-1}
\end{equation}%
We assume that the loadings $\Lambda $ satisfy%
\begin{equation*}
\Lambda =(\Lambda _{1}^{\prime },...,\Lambda _{p_{2}}^{\prime })^{\prime },\ %
\mbox{satisfying}\ \Lambda _{i}^{\prime }\Lambda _{j}/p_{1}=0\ i\neq j,\ %
\mbox{and}\ \Lambda _{i}^{\prime }\Lambda _{i}/p_{1}=I_{k};
\end{equation*}%
that is, the blocks $\Lambda _{1},...,\Lambda _{p_{2}}$ of $\Lambda $ are
orthonormal, and their columns span a $p_{2}$-dimensional linear space. In
this setting, it holds that $\Lambda ^{\prime }\Lambda /(p_{1}p_{2})=I_{k}$,
which satisfies the strong/pervasive factor condition for the vector factor
model. Equation (\ref{vec-kron-1}) can be rewritten artificially (i.e.,
without meaningful row and column cross-sections) in matrix form, viz.
\begin{equation}
X_{t}=(\Lambda _{1}F_{t},...,\Lambda _{p_{2}}F_{t})+E_{t},
\label{vec-kron-2}
\end{equation}%
where the errors are such that $E\left( E_{t}^{\prime }E_{t}\right) /p_{1}$
has bounded eigenvalues, as also stipulated by our Assumption \ref%
{idiosyncratic}; we assume for simplicity that $E\left( F_{t}^{\prime
}F_{t}\right) =c_{0}$. Then, by standard algebra, it follows that the column
covariance matrix is given by
\begin{equation*}
{\Sigma }_{c}=E(\Lambda _{1}{F}_{t},...,\Lambda _{p_{2}}{F}_{t})^{\prime
}(\Lambda _{1}{F}_{t},...,\Lambda _{p_{2}}{F}_{t})/p_{1}=c_{0}I_{p_{2}}.
\end{equation*}%
Thus, the second moment matrix of the signal $(\Lambda _{1}F_{t},...,\Lambda
_{p_{2}}F_{t})$ has eigenvalues of the same order of magnitude as the second
moment matrix of the idiosyncratic errors, which entails that there are no
strong (or even weak) column factors. Similarly, considering the row
covariance matrix and assuming, again for simplicity, that $E\left(
F_{t}F_{t}^{\prime }\right) =c_{1}I_{k}$, it holds that
\begin{equation*}
{\Sigma }_{r}=E(\Lambda _{1}{F}_{t},...,\Lambda _{p_{2}}{F}_{t})(\Lambda _{1}%
{F}_{t},...,\Lambda _{p_{2}}{F}_{t})^{\prime }/p_{1}=c_{1}\left(
\sum_{i=1}^{p_{2}}\Lambda _{i}\Lambda _{i}^{\prime }\right) /p_{1}.
\end{equation*}%
It is easy to see that this matrix is idempotent,\footnote{%
Indeed, it holds that%
\begin{equation*}
\left( \frac{1}{p_{1}}\sum_{i=1}^{p_{2}}\Lambda _{i}\Lambda _{i}^{\prime
}\right) \left( \frac{1}{p_{1}}\sum_{i=1}^{p_{2}}\Lambda _{i}\Lambda
_{i}^{\prime }\right) =\frac{1}{p_{1}^{2}}\sum_{i,j=1}^{p_{2}}\Lambda
_{i}\Lambda _{i}^{\prime }\Lambda _{j}\Lambda _{j}^{\prime }=\frac{1}{%
p_{1}^{2}}\sum_{i=1}^{p_{2}}\Lambda _{i}\Lambda _{i}^{\prime }\Lambda
_{i}\Lambda _{i}^{\prime }=\frac{1}{p_{1}}\sum_{i=1}^{p_{2}}\Lambda
_{i}\Lambda _{i}^{\prime };
\end{equation*}%
} and therefore its eigenvalues belong in $\{0,1\}$; hence, the row
covariance matrix $p_{1}{\Sigma }_{r}/p_{2}$ has eigenvalues $0$ and $%
p_{1}/p_{2}$. If $p_{1}/p_{2}$ is bounded, the largest eigenvalues of the
row signal matrix ${\Sigma }_{r}$ are all bounded, which again leads to
finding no strong (or even weak) row factors. Thus, we conclude that a
methodology based on assuming (\ref{kron}) incorrectly detects no common
factors, either in the columns or in the rows, despite the existence of a
factor structure in the vector factor model.

In general, as discussed in the introduction, (\ref{kron}) is likely to be
an adequate model where there is a meaningful matrix structure, with
economically meaningful row and column cross-sections; the examples
discussed in the introduction are some of the possible cases in which (\ref%
{kron}) is a natural formulation, based on the very nature of the data. This
said, as a preliminary step in the analysis, it would nonetheless be highly
desirable to have a formal test to check whether a Kronecker product
structure does exist in the loadings space $\Lambda $. We are not aware of
any such test in a high-dimensional context. A recent contribution by %
\citet{guggenb} provides a test for the null of having a Kronecker product
structure in a fixed dimensional matrix; similarly, \citet{chenvar} propose
a test for the Kronecker product structure in the context of a Vector
AutoRegressive model for a matrix-valued time series, but in that case the
row and column dimensions of the data $X_{t}$ are both fixed. The highly
nontrivial problem of developing a test for a Kronecker product structure in
the large dimensional case is currently under investigation by the authors.

\section*{Acknowledgements}

He's work is supported by NSF China (12171282,11801316), National
Statistical Scientific Research Key Project (2021LZ09), Project funded by
China Postdoctoral Science Foundation (2021M701997) and the Fundamental
Research Funds of Shandong University, Young Scholars Program of Shandong
University, China. Kong's work is partially supported by NSF China (71971118
and 11831008) and the WRJH-QNBJ Project and Qinglan Project of Jiangsu
Province. The authors would like to thank the Editor Xiaohong Chen, an
anonymous Associate Editor, and three anonymous Referees, whose helpful
comments have greatly improved the quality and focus of the paper.

\section*{Supplementary Material}

Further discussions on the assumptions, the technical proofs of the main
results and extra simulation and empirical studies are included in the
Supplementary Material.

\begin{adjustwidth}{-5pt}{-5pt}

{\footnotesize {\
\bibliographystyle{chicago}
\bibliography{LTbiblio}
} }

\end{adjustwidth}

\end{document}